%% file: main.tex
\documentclass[letterpaper,11pt]{article}
\input{defs}

\usepackage[margin=1in]{geometry}
\usepackage{fullpage}
\usepackage{microtype}
\usepackage{silence}
\DeclareMathOperator{\vol}{\mathbf{vol}}

\newcommand*\samethanks[1][\value{footnote}]{\footnotemark[#1]}

\title{Dynamic Connectivity, Minimum Spanning Tree, and 2-Edge Connectivity with Polylogarithmic Worst-Case Update Time}

\author{Simon Meierhans\thanks{The research leading to these results has received funding from grant no. 200021 204787 of the Swiss National Science Foundation. Simon Meierhans was supported by a Google PhD Fellowship.} \\
ETH Zurich \\
mesimon@inf.ethz.ch 
\and  
Maximilian Probst Gutenberg\samethanks[1] \\
ETH Zurich \\
maximilian.probst@inf.ethz.ch
\and
Yu-Cheng Yeh \\
ETH Zurich \\
yuyehy@student.ethz.ch
}
\date{}

\begin{document}
\pagenumbering{gobble}

\maketitle

\input{abstract}
\vfill

\pagebreak
\tableofcontents
\pagebreak
\pagenumbering{arabic}

\input{intro}
\input{overview}
\input{preliminaries}
\input{MST}
\input{2EdgeConnectivity}
\input{coreGraphFramework}

\newpage
\bibliographystyle{alpha}
\bibliography{refs}

\newpage
\appendix
\crefalias{section}{appendix}
\input{appendix}

\end{document}

%% file: defs.tex
\usepackage[dvipsnames,svgnames,table]{xcolor}
\usepackage{enumitem}
\usepackage{graphicx}
\usepackage{fancybox}
\usepackage{comment}
\usepackage{xcolor}
\usepackage{nameref}
\usepackage[noadjust]{cite}
\usepackage{bbm}
\usepackage{soul}
\usepackage{mathtools}
\usepackage{tikz}
\usetikzlibrary{positioning, backgrounds, arrows.meta}

\definecolor{ForestGreen}{rgb}{0.1333,0.5451,0.1333}
\definecolor{DarkRed}{rgb}{0.8,0,0}
\definecolor{Red}{rgb}{1,0,0}
\usepackage[linktocpage=true,
pagebackref=true,colorlinks,
linkcolor=ForestGreen,citecolor=ForestGreen,
bookmarks,bookmarksopen,bookmarksnumbered]
{hyperref}
\usepackage[nottoc,numbib]{tocbibind}

\usepackage{amsmath}
\usepackage{amssymb}
\usepackage{amsthm}

\usepackage[ruled, noend, linesnumbered]{algorithm2e}

\usepackage{subcaption}

\usepackage{thm-restate}

\usepackage{float}
\usepackage{cleveref}

\newtheorem{theorem}{Theorem}[section]

\newtheorem{corollary}[theorem]{Corollary}
\newtheorem{lemma}[theorem]{Lemma}
\newtheorem{observation}[theorem]{Observation}

\newtheorem{claim}[theorem]{Claim}

\newtheorem{fact}[theorem]{Fact}

\newtheorem{definition}[theorem]{Definition}

\newtheorem*{theorem*}{Theorem}
\newtheorem*{corollary*}{Corollary}
\newtheorem*{conjecture*}{Conjecture}
\newtheorem*{lemma*}{Lemma}
\newtheorem*{thm*}{Theorem}
\newtheorem*{prop*}{Proposition}
\newtheorem*{obs*}{Observation}
\newtheorem*{definition*}{Definition}

\newtheorem*{remark*}{Remark}
\newtheorem*{rec*}{Recommendation}

\newenvironment{fminipage}%
  {\begin{Sbox}\begin{minipage}}%
  {\end{minipage}\end{Sbox}\fbox{\TheSbox}}

\def\defeq{\stackrel{\mathrm{def}}{=}}

\def\ceil#1{\left\lceil #1 \right\rceil}

\newcommand\ww{\boldsymbol{\mathit{w}}}

\renewcommand{\deg}{\operatorname{deg}}

\newcommand\R{\mathbb{R}}

\DeclareMathOperator{\msf}{MSF}

\newcommand{\eps}{\varepsilon}

\renewcommand{\tilde}{\widetilde}

\DeclareFontFamily{U}{mathb}{\hyphenchar\font45}
\DeclareFontShape{U}{mathb}{m}{n}{<5> <6> <7> <8> <9> <10> gen * mathb
<10.95> mathb10 <12> <14.4> <17.28> <20.74> <24.88> mathb12}{}
\DeclareSymbolFont{mathb}{U}{mathb}{m}{n}
\DeclareMathSymbol{\rcirclearrow}{\mathbin}{mathb}{'367}

\renewcommand{\bar}{\overline}

\newif\ifrandom
\randomtrue

\newcommand{\todolater}[1]{}

\DeclareUnicodeCharacter{2113}{$\ell$}

%% file: abstract.tex
\begin{abstract}
We give fully dynamic algorithms for maintaining connectivity, minimum spanning tree, and $2$-edge connectivity of a graph with worst-case polylogarithmic update time. Our algorithms are randomized and succeed with high probability against an adaptive adversary. For the minimum spanning tree and $2$-edge connectivity problems, this improves over the subpolynomial update time bounds obtained by Nanongkai, Saranurak, and Wulff-Nilsen [FOCS'17], Jin and Sun [FOCS'21], and Jin, Sun, and Thorup [SODA'24], respectively.

The only randomized component of our algorithms is the computation of static expander decompositions, and a deterministic algorithm for said problem would directly imply deterministic algorithms for all three problems. This reduction is novel even for the connectivity problem.
\end{abstract}

%% file: intro.tex
\section{Introduction}
\label{sec:intro}

In this article, we study three fundamental problems on dynamic graphs:
\begin{itemize}
    \item \emph{Connectivity:} At its most basic, the connectivity problem asks us to determine if a graph is connected or not. This was the first problem ever to be studied on dynamic graphs, and has been the subject of a long line of research \cite{spira1975finding, chin1978algorithms, harel1982line, frederickson1985data, galil1991fully, HenzingerKing95, holm01, thorup2000near, patrascu2007planning, kapron, WulffNilsen13,wang2015improvedrandomizeddatastructure, gibb2015dynamicgraphconnectivityimproved, kejlberg2016faster, NS17, W17, NSW17, chuzhoy2020deterministic, expander_hierarchy, huang2023fully, meierhans2026dynamic}. Algorithms for the connectivity problem maintain a maximal spanning forest, and the main challenge is to find a replacement edge, or certify that there is none, whenever a forest edge gets deleted. 
    \item \emph{Minimum Spanning Tree (MST):} The minimum spanning forest of a graph is the maximal spanning forest with minimum total weight. By adding $\infty$-weight edges, one can assume w.l.o.g. that $G$ is connected; thus, the problem is typically referred to as the \emph{MST} problem. 
    MSTs are among the first graph-theoretic objects taught to undergraduate students, and the problem has a rich history that is intimately connected to the connectivity problem \cite{frederickson1985data, eppstein90planarmst, EPPSTEIN1994237, it97, holm01, holm2015faster, W17, NSW17}.
    \item \emph{2-Edge Connectivity:} The $2$-edge connectivity problem instead seeks to identify maximal components that stay connected even after removing an edge \cite{westbrook1992bridge, galil1991fully, frederickson1997ambivalent, it97, gibb2015dynamicgraphconnectivityimproved, holm2017bridge, jin2021fully, jin2024fully}. If removing a single edge disconnects a component, we call the edge a \textit{bridge}. This problem has been studied extensively and is usually solved by maintaining a maximal spanning forest, and then additionally recording which forest edges are bridges.
\end{itemize}

Despite the importance of these problems, the fastest algorithms for minimum spanning tree and $2$-edge connectivity still require sub-polynomial worst-case update time. In this article, we resolve this problem by providing $\tilde{O}(1)$ worst-case update time algorithms for all three problems.\footnote{In this article, we use $\tilde{O}(\cdot)$ notation to suppress poly-logarithmic factors in the number of vertices of the graph.}

\begin{restatable}{theorem}{mainThm} 
    \label{thm:main}
    There are fully-dynamic graph algorithms for the Connectivity, Minimum Spanning Tree, and 2-Edge Connectivity problems with $\tilde{O}(1)$ worst-case update time. All algorithms are randomized and withstand an adaptive adversary, and they succeed w.h.p. at each update. For the Connectivity and Minimum Spanning Tree problems, the algorithms report failure.
\end{restatable}

Furthermore, our algorithms only use randomness to compute static expander decompositions in nearly-linear time; a deterministic algorithm for this problem thus immediately de-randomizes them. Such a reduction to deterministic expander decomposition is novel even for the connectivity problem.  

We next review previous work for fully-dynamic Connectivity, MST, and $2$-Edge Connectivity, and describe fundamental roadblocks for algorithms with poly-logarithmic worst-case update time. In \Cref{sec:overview}, we give an overview of our algorithm and describe how we overcome these challenges.

\subsection{Previous Work}

\paragraph{Connectivity.}

Algorithms for the dynamic connectivity problem maintain a maximal spanning forest. After an edge deletion, we are tasked with finding a replacement edge across the resulting cut or certifying that no such edge exists. Two algorithms are currently known to achieve poly-logarithmic worst-case update time: a Monte-Carlo algorithm by \cite{kapron} that uses uniform sub-sampling and clever bit tricks for finding an edge, and a Las-Vegas algorithm by \cite{meierhans2026dynamic} that reduces the problem to maintaining a spanning tree on an expander, where replacement edges can be found quickly by uniformly sampling edges incident to the smaller tree component. The former algorithm heavily leverages the oblivious adversary assumption, while the latter withstands an adaptive adversary. It remains a major open question to obtain an algorithm with deterministic guarantees. However, both of these algorithms are difficult to de-randomize. In contrast, our novel approach does not require subsampling and only requires randomization to compute a static expander decomposition, one of the key primitives in modern graph algorithms. 

\paragraph{Minimum Spanning Tree.} 

For the MST problem, it seems much harder to leverage randomness to find a replacement edge across a cut, since a specific edge must be identified, namely the one with the smallest weight. The amortized algorithm in \cite{holm01} checks all edges incident to vertices on the smaller side of the cut, but then promotes them to build a hierarchy that prevents extensive searches from occurring frequently. In a major breakthrough \cite{W17, NS17, NSW17}, a dynamic minimum spanning tree algorithm was given, achieving sub-polynomial worst-case update time. The algorithm maintains a hierarchy based on expander decompositions. We give a novel way to leverage a key insight by \cite{W17} to design an alternative hierarchy. While we still rely on expander decompositions, our new approach allows us to decouple their usage from the hierarchy. This allows us to obtain very low recourse, which is necessary to avoid snowballing behaviour that forces sub-logarithmic hierarchy levels, typically resulting in sub-polynomial update times.

\paragraph{2-Edge Connectivity.}
Given an algorithm for fully dynamic Connectivity/MST, one can easily query whether an edge in a maintained spanning tree is a bridge by simply removing it and then checking if a replacement edge is found. However, this strategy is not suitable for quickly answering whether the whole graph is $2$-edge connected. Instead, the seminal results \cite{frederickson1997ambivalent, holm01} suggest embedding edges into a spanning tree $T$, i.e., for each off-tree edge $f = xy$, all edges on the $x$-to-$y$ path in $T$ are marked as covered. It is well-known that a graph is $2$-edge connected if all edges of the tree $T$ are covered.

The key challenge with this approach is that some edges $e$ in $T$ might have high congestion, i.e., $\Omega(m)$ edges $f$ might cover $e$. Now, when $e$ is deleted from $T$ and replaced by another edge $e'$ that previously covered $e$, we have to re-embed all edges $f$ that previously covered $e$, causing high update times. While \cite{holm01} designs a clever approach to amortize such events, their strategy can not mitigate the worst case.

Towards obtaining worst-case guarantees, \cite{chalermsook2021vertex} proposed to build a vertex-sparsifier that removes redundant coverage information. This sparsifier is exploited in \cite{jin2021fully, jin2024fully} to give a sub-polynomial worst-case update time algorithm for the $c$-edge connectivity for all constants $c > 0$. However, again, their vertex-sparsifier is computed directly on an expander decomposition, which can cause (super-)logarithmic recourse at each level. Thus, they need to keep the number of layers of their hierarchy sub-logarithmic, yielding sub-polynomial worst-case update time.

Here, we give an alternative approach that combines and extends these ideas: we decompose the graph into expanders and maintain coverage within the expanders. In the case of high congestion, we do not need to maintain the coverage explicitly, since expanders guarantee that such components are stable. This allows us to maintain the $2$-edge connectivity information of all edges in the expander. Finally, the coverage information of edges that are in the boundary of the expander decomposition can be deferred to the next level of the hierarchy. Thus, instead of having a few changes to the vertex sparsifier after each update, we have a few additional edges whose coverage information is deferred until later. While the change is subtle, we show that the latter changes can be re-incorporated with poly-logarithmic worst-case update time.

\paragraph{Beyond Fully-Dynamic Connectivity, MST and $2$-Edge Connectivity.}

As mentioned above, \cite{jin2021fully} gives an algorithm with sub-polynomial worst-case update time to maintain pairwise $c$-edge connectivity, i.e., for any pair of vertices $u,v$ they can return whether there are at least $c$ edge-disjoint paths between $u$ and $v$ in the underlying graph. While previously, $2$- and $3$-edge connectivity were studied \cite{HenzingerKing95, frederickson1997ambivalent, galil1992fully}, no sub-polynomial worst-case update time algorithm was known.

In the dynamic global mincut problem, the graph has to report the size of the global mincut, i.e., the smallest number of edges that disconnect it. While this problem has been previously studied \cite{thorup2007fully, jin2024fully, de2025tree}, only very recently, an algorithm with sub-polynomial worst-case update time \cite{el2025fully, el2026deterministic} was given for the case where the mincut is of sub-polynomial size, or for returning a $(1+\eps)$-approximation. 

In adjacent works, the biconnectivity problem was studied \cite{henzinger1995fully, henzinger1995certificates, henzinger2000improved, holm01, holm2025fully}, where the algorithm has to output whether there is a vertex in the underlying graph whose removal disconnects the graph. While amortized algorithms with poly-logarithmic update time were given in \cite{holm01, holm2025fully}, it remains an open problem to find an algorithm even with sub-polynomial worst-case update time.

%% file: overview.tex
\section{Overview} 
\label{sec:overview}

\paragraph{Dynamic Algorithms via the Core Graph Framework.}

The dynamic core graph framework suggests interleaving dynamic edge and vertex sparsification to build dynamic algorithms. The first layer might consist of an $n$-vertex graph with $m = \tilde{O}(n)$ edges, the second then works on a graph with at most $n/k$ vertices and $\tilde{O}(n/k)$ edges, and so on. This naturally leads to a hierarchy of graphs with geometrically decreasing sizes, which enables us to re-initialize the smaller graphs more frequently.\footnote{To obtain worst-case update times, one needs to de-amortize re-initializations; however, this can be achieved by standard techniques.} 

For connectivity, recent work \cite{meierhans2026dynamic} has shown vertex sparsification to be rather straightforward, and we observe that their algorithm can be naturally extended to give vertex sparsifiers also for the minimum spanning forest and $2$-edge connectivity problems. We thus focus on presenting our novel edge sparsification algorithms in this overview. 

We next outline our dynamic edge sparsifier for MST in \Cref{subsec:overview_mst_spars} and our dynamic edge sparsifier for 2-edge connectivity in \Cref{subsec:overview_2edge_spars}. The full technical details for the edge sparsifiers are given in \Cref{sec:mst} and \Cref{sec:2edgeconn}, respectively. In \Cref{sec:mainSec}, we show that given these edge sparsifiers as subroutines, it is rather straightforward to develop algorithms that solve the respective problems.

\subsection{Dynamic Edge Sparsifier for Minimum Spanning Tree}
\label{subsec:overview_mst_spars}

\paragraph{Dynamic MST Sparsifiers.}

Let us start by defining the notion of an edge sparsifier for the minimum spanning tree problem. We call a subgraph $H \subseteq G$ a \textit{minimum spanning forest (MSF) sparsifier} if the minimum spanning forest of $H$ and $G$ coincide. Since a minimum spanning tree itself can be computed extremely fast, this concept is superfluous for static graphs. For dynamic graphs, however, maintaining an MST sparsifier with few edges seems very tricky at first.

One might first be tempted to try to extend the recent edge sparsifier for connectivity from Meierhans and Probst Gutenberg \cite{meierhans2026dynamic} or to exploit the techniques by Kapron, King, and Mountjoy \cite{kapron} that have previously been the key to poly-logarithmic worst-case update time for the connectivity problem. But both algorithms attempt to reconnect two connected components $A, B$ of $H$ after the deletion of an edge $f$ with endpoints in $A$ and $B$ by sampling uniformly from the cut $E_G(A, B)$. \cite{kapron} relies on a clever bit trick and the obliviousness of the adversary; \cite{meierhans2026dynamic} reduces to expanders and can therefore assume $|E_G(A, B)| = \tilde{\Omega}(\vol_G(A))$, thus sampling edges incident to $A$ uniformly at random yields a cut edge after only $\tilde{O}(1)$ samples with high probability.

However, note that for the MST problem, we are asking for the edge in $E_G(A, B)$ of \textbf{minimum} weight, not just an arbitrary edge. Since $|E_G(A, B)|$ can be rather large, previous strategies do not extend, and a novel approach needs to be designed.

While it is natural to consider bucketing by edge weights, a naive approach -- bucketing all edges with weight in $[(1+\eps)^i, (1+\eps)^{i+1})$ and maintaining connectivity sparsifiers on these edge sets -- yields sparsifiers that only approximate the MST in weight by a $(1+\eps)$-factor. This provides no improvement over the approximate minimum spanning tree obtained from a connectivity algorithm that returns a spanning forest in a black-box manner. 

\paragraph{Peeling Non-important High-Weight Edges.} 

In our approach, we also partition the edges by weight; however, we use a much more careful algorithm to ensure that the exact MST remains contained in $H$. 

Our approach starts with the following fundamental, but crucial, insight by Wulff-Nilsen \cite{W17}: if we have two buckets $E_\text{low}$ and $E_\text{high}$ containing the low-weight and high-weight edges respectively, then an edge in $E_\text{high}$ can only be in the MST if it crosses between two connected components in the graph $G[E_\text{low}]$. Below is a more formal version of this statement. 

\begin{observation}[\cite{W17}] \label{obs:wn}
    Let $G = (V, E, \ww)$ be a weighted graph, and let $E_{\textnormal{low}}, E_{\textnormal{high}}$ be a partition of $E$ such that $\ww(e) \leq \ww(f)$ for all $e \in E_{\textnormal{low}}, f \in E_{\textnormal{high}}$. Then the minimum spanning tree of $G$ is in $(V, E_{\textnormal{low}} \cup E_{\textnormal{cross}})$ where \begin{equation*}
        E_{\textnormal{cross}} = \{uv \in E_{\textnormal{high}} \mid u,v \textnormal{ in separate connected components of } (V, E_{\textnormal{low}}) \}.
    \end{equation*}
\end{observation}

In \cite{W17} and the subsequent article \cite{NSW17}, the insight above is combined with the expander decomposition framework to produce a collection of smaller sub-problems, onto which they recurse and then combine into a single MSF. However, it is not clear how their rather intricate algorithm can be sped up to have polylogarithmic update time.

Instead, we follow the same insight but give a simple and fast algorithm: we first peel a small fraction of high-weight edges, and then distinguish between two cases. Either this breaks the graph into a bunch of small connected components (the balanced case), or there is a large component after removing the high-weight edges (the unbalanced case).

\paragraph{The Balanced Case.}

We let $E_{\text{high}}$ consist of a $\gamma = c/\log_2 n$ fraction of the edges for a small constant $c > 0$, i.e., $E_{\text{high}}$ contains the $\gamma \cdot |E|$ highest weight edges in $E$. Assume that removing $E_{\text{high}}$ causes the graph to shatter into components $\mathcal{X} = X_1, \ldots X_k$ such that $|X_i| \leq |V|/2$ for all $i$. If this assumption holds repeatedly, iterating this procedure by applying it to each connected component yields singleton vertices after $\log_2 n$ iterations. Then, simply adding all edges that ever end up in a set $E_{\text{high}}$ and cross two components directly yields an MST sparsifier, and since the recursion depth is only $\log_2 n$, at most a $c$ fraction of the edges of the initial graph make their way into the sparsifier. See \Cref{fig:balanced_case} for an illustration of this case. 

However, always reaching this scenario is not only not realistic, but impossible. Each such round increases the density of the remaining connected components, and sooner or later, there is going to be a large connected component $|X_j| \geq |V|/2$. In the next paragraph, we outline our strategy for this crucial case. 

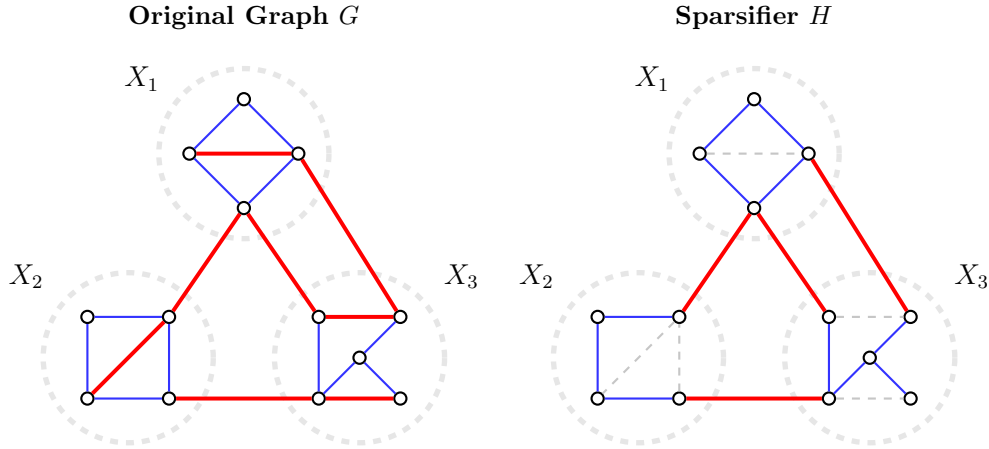
\begin{figure}[H]
        \centering
        \begin{tikzpicture}[
            scale=0.9,
            every node/.style={transform shape},
            vertex/.style={circle, draw=black, fill=white, thick, inner sep=1.8pt},
            blob/.style={draw=gray!20, dashed, line width=2pt},
            heavy edge/.style={red, line width=1.5pt},
            light edge/.style={blue!80, line width=0.8pt},
            discarded/.style={gray!45, dashed, line width=0.8pt}
        ]
    
        \def\placenodes#1{
            \node[vertex] (#1-t1) at (0, 3.1) {};
            \node[vertex] (#1-t2) at (-0.8, 2.3) {};
            \node[vertex] (#1-t3) at (0.8, 2.3) {};
            \node[vertex] (#1-t4) at (0, 1.5) {};
    
            \node[vertex] (#1-l1) at (-1.1, -0.1) {};
            \node[vertex] (#1-l2) at (-2.3, -0.1) {};
            \node[vertex] (#1-l3) at (-1.1, -1.3) {};
            \node[vertex] (#1-l4) at (-2.3, -1.3) {};
    
            \node[vertex] (#1-r1) at (1.1, -0.1) {};
            \node[vertex] (#1-r2) at (2.3, -0.1) {};
            \node[vertex] (#1-r3) at (1.1, -1.3) {}; 
            \node[vertex] (#1-r4) at (2.3, -1.3) {};
            \node[vertex] (#1-r5) at (1.7, -0.7) {};
        }
    
        \def\drawblobs{
            \begin{scope}[on background layer]
                \draw[blob] (0, 2.3) circle (1.25);
                \node[font=\large, text=black] at (-1.5, 3.4) {$X_1$};
                
                \draw[blob] (-1.7, -0.7) circle (1.25);
                \node[font=\large, text=black] at (-3.2, 0.5) {$X_2$};
                
                \draw[blob] (1.7, -0.7) circle (1.25);
                \node[font=\large, text=black] at (3.2, 0.5) {$X_3$};
            \end{scope}
        }
    
        \begin{scope}[shift={(0,0)}]
            \node[font=\bfseries, align=center] at (0, 4.3) {Original Graph $G$};
            \placenodes{A}
            \drawblobs
    
            \draw[heavy edge] (A-t2) -- (A-t3);
            \draw[light edge] (A-t1) -- (A-t2);
            \draw[light edge] (A-t1) -- (A-t3);
            \draw[light edge] (A-t2) -- (A-t4);
            \draw[light edge] (A-t3) -- (A-t4);
    
            \draw[heavy edge] (A-l1) -- (A-l4);
            \draw[light edge] (A-l1) -- (A-l2);
            \draw[light edge] (A-l2) -- (A-l4);
            \draw[light edge] (A-l4) -- (A-l3);
            \draw[light edge] (A-l3) -- (A-l1);
    
            \draw[heavy edge] (A-r1) -- (A-r2);
            \draw[heavy edge] (A-r3) -- (A-r4);
            \draw[light edge] (A-r1) -- (A-r3);
            \draw[light edge] (A-r2) -- (A-r5);
            \draw[light edge] (A-r4) -- (A-r5);
            \draw[light edge] (A-r3) -- (A-r5);
            
            \draw[heavy edge] (A-t4) -- (A-l1);
            \draw[heavy edge] (A-t4) -- (A-r1);
            \draw[heavy edge] (A-t3) -- (A-r2);
            \draw[heavy edge] (A-l3) -- (A-r3);
        \end{scope}
    
        \begin{scope}[shift={(7.5,0)}]
            \node[font=\bfseries, align=center] at (0, 4.3) {Sparsifier $H$};
            \placenodes{B}
            \drawblobs
    
            \draw[discarded] (B-t2) -- (B-t3);
            \draw[light edge] (B-t1) -- (B-t2);
            \draw[light edge] (B-t1) -- (B-t3);
            \draw[light edge] (B-t2) -- (B-t4);
            \draw[light edge] (B-t3) -- (B-t4);
    
            \draw[discarded] (B-l1) -- (B-l4);
            \draw[discarded] (B-l3) -- (B-l1);
            \draw[light edge] (B-l1) -- (B-l2);
            \draw[light edge] (B-l2) -- (B-l4);
            \draw[light edge] (B-l4) -- (B-l3);
    
            \draw[discarded] (B-r1) -- (B-r2);
            \draw[discarded] (B-r3) -- (B-r4);
            \draw[light edge] (B-r1) -- (B-r3); 
            \draw[light edge] (B-r2) -- (B-r5);
            \draw[light edge] (B-r4) -- (B-r5);
            \draw[light edge] (B-r3) -- (B-r5);
    
            \draw[heavy edge] (B-t4) -- (B-l1);
            \draw[heavy edge] (B-t4) -- (B-r1);
            \draw[heavy edge] (B-t3) -- (B-r2);
            \draw[heavy edge] (B-l3) -- (B-r3);
        \end{scope}
    
        \end{tikzpicture}
        \caption{The \emph{balanced} case. The edges in $E_{\text{high}}$ are displayed in red, and the remaining edges in $E \setminus E_{\text{high}}$ are displayed in blue. Notice that all inter-component edges are preserved, and each red edge in a cluster is unused by the sparsifier.}
        \label{fig:balanced_case}
\end{figure}

\paragraph{The Unbalanced Case.}

If the graph does not get shattered into small components, we are instead going to recurse on the largest component, which we call $X_{\text{large}}$, and everything else $X_{\text{small}} = V \setminus X_{\text{large}}$. However, most or all the edges in $E_{\text{high}}$ could be in $E_\text{cross} = E(X_{\text{small}}, X_{\text{large}})$. We therefore make a crucial observation that allows us to sparsify the cut. For every vertex $v \in X_{\text{small}}$, we let $e_v$ be the lowest weight edge in $E(\{v\}, X_{\text{large}})$, i.e., the lowest weight edge from $v$ to the component $X_{\text{large}}$. These edges $\{e_v \mid v \in X_{\text{small}} \}$ contain all edges of $E_\text{high}$ across the cut $(X_{\text{small}}, X_{\text{large}})$ that can end up in the MST. 

To see why this is the case, assume, towards a contradiction, that two such edges $xu$ and $xv$ in the MST are adjacent to a vertex $x \in X_{\text{small}}$, and assume $\ww(xu) > \ww(xv)$ w.l.o.g since we can assume edge weights to be distinct. By construction, $G[X_{\text{large}}] \setminus E_{\text{high}}$ is connected, so there is a path from $u$ to $v$ where every edge has weight less than $\ww(xu)$. But then $xu$ is the highest-weight edge in a cycle of $G$, which contradicts its being in the MST, because it can be swapped with an edge that preserves connectivity and decreases the total weight. See \Cref{fig:unbalanced_case} for an illustration. 

Crucially, every vertex can be on the smaller side of such a cut at most $\log_2 n$ times, and we therefore only add $n \log n$ additional edges to the sparsifier due to the unbalanced case. 

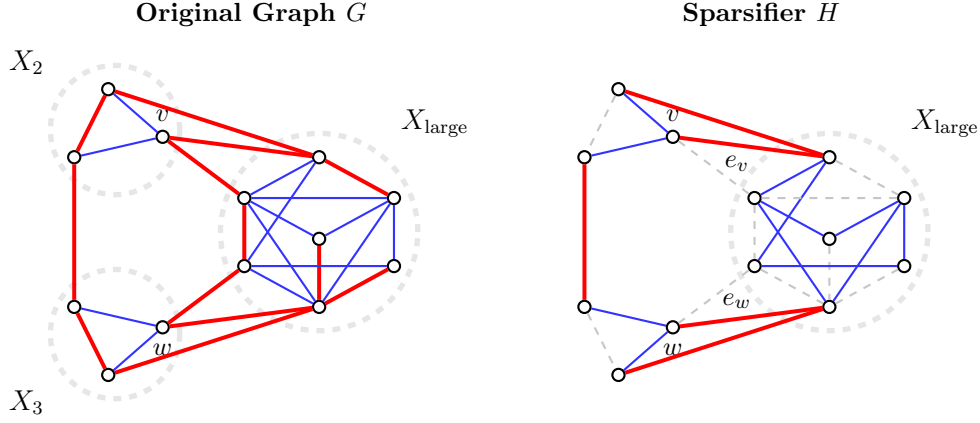
\begin{figure}[H]
    \centering
    \begin{tikzpicture}[
        scale=0.9,
        every node/.style={transform shape},
        vertex/.style={circle, draw=black, fill=white, thick, inner sep=1.8pt},
        blob/.style={draw=gray!20, dashed, line width=2pt},
        heavy edge/.style={red, line width=1.5pt},
        light edge/.style={blue!80, line width=0.8pt},
        discarded/.style={gray!45, dashed, line width=0.8pt}
    ]

    \def\placenodes#1{
        \node[vertex] (#1-u1) at (1.5, 2.1) {};
        \node[vertex] (#1-u2) at (0.4, 1.5) {};
        \node[vertex] (#1-u3) at (0.4, 0.5) {};
        \node[vertex] (#1-u4) at (1.5, -0.1) {};
        \node[vertex] (#1-u5) at (2.6, 0.5) {};
        \node[vertex] (#1-u6) at (2.6, 1.5) {};
        \node[vertex] (#1-u7) at (1.5, 0.9) {};

        \node[vertex, label={[font=\normalsize, text=black]90:$v$}] (#1-v1) at (-0.8, 2.4) {};
        \node[vertex] (#1-v2) at (-1.6, 3.1) {};
        \node[vertex] (#1-v3) at (-2.1, 2.1) {};

        \node[vertex, label={[font=\normalsize, text=black]-90:$w$}] (#1-w1) at (-0.8, -0.4) {};
        \node[vertex] (#1-w2) at (-2.1, -0.1) {};
        \node[vertex] (#1-w3) at (-1.6, -1.1) {};
    }

    \begin{scope}[shift={(0,0)}]
        \node[font=\bfseries, align=center] at (0.5, 4.2) {Original Graph $G$};
        \placenodes{A}
        
        \begin{scope}[on background layer]
            \draw[blob] (1.5, 1.0) circle (1.45);
            \node[font=\large, text=black] at (3.2, 2.6) {$X_{\text{large}}$};
            
            \draw[blob] (-1.5, 2.5) circle (0.95);
            \node[font=\large, text=black] at (-2.8, 3.5) {$X_2$};
            
            \draw[blob] (-1.5, -0.5) circle (0.95);
            \node[font=\large, text=black] at (-2.8, -1.5) {$X_3$};
        \end{scope}

        \draw[heavy edge] (A-u1) -- (A-u6);
        \draw[heavy edge] (A-u2) -- (A-u3);
        \draw[heavy edge] (A-u4) -- (A-u5);
        \draw[heavy edge] (A-u7) -- (A-u4);
        \draw[light edge] (A-u1) -- (A-u2);
        \draw[light edge] (A-u2) -- (A-u4);
        \draw[light edge] (A-u4) -- (A-u6);
        \draw[light edge] (A-u6) -- (A-u5);
        \draw[light edge] (A-u5) -- (A-u3);
        \draw[light edge] (A-u3) -- (A-u1);
        \draw[light edge] (A-u2) -- (A-u6);
        \draw[light edge] (A-u3) -- (A-u4);
        \draw[light edge] (A-u2) -- (A-u7);
        \draw[light edge] (A-u6) -- (A-u7);

        \draw[heavy edge] (A-v2) -- (A-v3);
        \draw[light edge] (A-v1) -- (A-v2);
        \draw[light edge] (A-v1) -- (A-v3);

        \draw[heavy edge] (A-w2) -- (A-w3);
        \draw[light edge] (A-w1) -- (A-w2);
        \draw[light edge] (A-w1) -- (A-w3);

        \draw[heavy edge] (A-v3) -- (A-w2);

        \draw[heavy edge] (A-v1) -- (A-u1);
        \draw[heavy edge] (A-v1) -- (A-u2);
        \draw[heavy edge] (A-w1) -- (A-u4);
        \draw[heavy edge] (A-w1) -- (A-u3);
        \draw[heavy edge] (A-v2) -- (A-u1);
        \draw[heavy edge] (A-w3) -- (A-u4);
    \end{scope}

    \begin{scope}[shift={(7.5,0)}]
        \node[font=\bfseries, align=center] at (0.5, 4.2) {Sparsifier $H$};
        \placenodes{B}
        
        \begin{scope}[on background layer]
            \draw[blob] (1.5, 1.0) circle (1.45);
            \node[font=\large, text=black] at (3.2, 2.6) {$X_{\text{large}}$};
        \end{scope}

        \draw[discarded] (B-u1) -- (B-u6);
        \draw[discarded] (B-u2) -- (B-u3);
        \draw[discarded] (B-u4) -- (B-u5);
        \draw[discarded] (B-u7) -- (B-u4);
        \draw[light edge] (B-u1) -- (B-u2);
        \draw[light edge] (B-u2) -- (B-u4);
        \draw[light edge] (B-u4) -- (B-u6);
        \draw[light edge] (B-u6) -- (B-u5);
        \draw[light edge] (B-u5) -- (B-u3);
        \draw[light edge] (B-u3) -- (B-u1);
        \draw[discarded] (B-u2) -- (B-u6);
        \draw[discarded] (B-u3) -- (B-u4);
        \draw[light edge] (B-u2) -- (B-u7);
        \draw[light edge] (B-u6) -- (B-u7);

        \draw[discarded] (B-v2) -- (B-v3);
        \draw[light edge] (B-v1) -- (B-v2);
        \draw[light edge] (B-v1) -- (B-v3);

        \draw[discarded] (B-w2) -- (B-w3);
        \draw[light edge] (B-w1) -- (B-w2);
        \draw[light edge] (B-w1) -- (B-w3);

        \draw[heavy edge] (B-v3) -- (B-w2);

        \draw[heavy edge] (B-v1) -- (B-u1) node[pos=0.4, below=0pt, text=black, font=\normalsize] {$e_v$};
        \draw[discarded] (B-v1) -- (B-u2);
        \draw[heavy edge] (B-w1) -- (B-u4) node[pos=0.4, above=0pt, text=black, font=\normalsize] {$e_w$};
        \draw[discarded] (B-w1) -- (B-u3);
        \draw[heavy edge] (B-v2) -- (B-u1);
        \draw[heavy edge] (B-w3) -- (B-u4);
    \end{scope}

    \end{tikzpicture}
    \caption{The \emph{unbalanced} case. The edges in $E_{\text{high}}$ are displayed in red, and the remaining edges in $E \setminus E_{\text{high}}$ are displayed in blue. Notice that in the sparsifier, every vertex in $X_\text{small}$ is connected to at most one vertex in $X_\text{large}$ using a red edge.}
    \label{fig:unbalanced_case}
\end{figure}

\paragraph{Dynamic Maintenance. }

Next, we discuss extending this framework to dynamic graphs. For a dynamic sparsifier, edge insertions can just be added to the sparsifier. We therefore focus on edge deletions. This is where expander components become important, because they are a stable version of connected components. 

Towards dynamizing this hierarchy, let us first make the strong simplifying assumption that the connected components that occur remain the same over the course of the update sequence. We will remove this assumption in the next paragraph. If this is the case, the only non-trivial maintenance is the lowest-weight edge connecting every vertex in $X_{\text{small}}$ to $X_{\text{large}}$. We can simply maintain this edge set with a priority queue at each vertex. Furthermore, the only edges that ever get deleted from the sparsifier edges are deleted in $G$, as desired. Next, we show that expander decompositions can be used to maintain stable connected components, which completes our algorithm. 

\paragraph{Expander Decomposition for Stable Connected Components.}

An expander decomposition partitions the vertex set of a graph $G = (V, E)$ into a collection of connected components $\mathcal{X} = \{X_1, \dots, X_k\}$ called expanders, such that at most $\rho \cdot m$ edges cross between components. Additionally, each such component is very well-connected, which makes it impossible to remove a small number of edges that disconnect a large number of vertices from a component $X_i$. This is exploited by expander pruning, which we use to remove an additional $\tilde{O}(1/\rho)$ edges after each deletion, such that every vertex that gets disconnected from $X_i$ becomes a singleton within $G[X_i]$ \cite{meierhans2025expanderpruningpolylogarithmicworstcase}. 

We use this machinery to maintain the sparsifier above, where we now compute an expander decomposition of $G \setminus E_{\text{high}}$. Then, the expanders take the place of the connected components, and we distinguish between the balanced and the unbalanced case as before. In the balanced case, we simply additionally insert the crossing edges of the expander decomposition into the sparsifier. In the unbalanced case, we also add them to the sparsifier, and do not consider these edges when we maintain the cheapest edge connecting vertices from $X_{\text{small}}$ to $X_{\text{large}}$. Choosing $\rho = 1/\tilde{O}(1)$ suffices to ensure that the extra edges added sum to $c' m$.  

\begin{figure}[H]
        \centering
        \begin{tikzpicture}[
            scale=0.9,
            every node/.style={transform shape},
            vertex/.style={circle, draw=black, fill=white, thick, inner sep=1.8pt},
            blob/.style={draw=gray!20, dashed, line width=2pt},
            heavy edge/.style={red, line width=1.5pt},
            light edge/.style={blue!80, line width=0.8pt},
            discarded/.style={gray!45, dashed, line width=0.8pt}
        ]
    
        \def\placenodes#1{
            \node[vertex] (#1-t1) at (0, 3.1) {};
            \node[vertex] (#1-t2) at (-0.8, 2.3) {};
            \node[vertex] (#1-t3) at (0.8, 2.3) {};
            \node[vertex] (#1-t4) at (0, 1.5) {};
    
            \node[vertex] (#1-l1) at (-1.1, -0.1) {};
            \node[vertex] (#1-l2) at (-2.3, -0.1) {};
            \node[vertex] (#1-l3) at (-1.1, -1.3) {};
            \node[vertex] (#1-l4) at (-2.3, -1.3) {};
    
            \node[vertex] (#1-r1) at (1.1, -0.1) {};
            \node[vertex] (#1-r2) at (2.3, -0.1) {};
            \node[vertex] (#1-r3) at (1.1, -1.3) {}; 
            \node[vertex] (#1-r4) at (2.3, -1.3) {};
            \node[vertex] (#1-r5) at (1.7, -0.7) {};
        }
    
        \def\drawblobs{
            \begin{scope}[on background layer]
                \draw[blob] (0, 2.3) circle (1.25);
                \node[font=\large, text=black] at (-1.5, 3.4) {$X_1$};
                
                \draw[blob] (-1.7, -0.7) circle (1.25);
                \node[font=\large, text=black] at (-3.2, 0.5) {$X_2$};
                
                \draw[blob] (1.7, -0.7) circle (1.25);
                \node[font=\large, text=black] at (3.2, 0.5) {$X_3$};
            \end{scope}
        }
    
        \begin{scope}[shift={(0,0)}]
            \node[font=\bfseries, align=center] at (0, 4.3) {Original Graph $G$};
            \placenodes{A}
            \drawblobs
    
            \draw[heavy edge] (A-t2) -- (A-t3);
            \draw[light edge] (A-t1) -- (A-t2);
            \draw[light edge] (A-t1) -- (A-t3);
            \draw[light edge] (A-t2) -- (A-t4);
            \draw[light edge] (A-t3) -- (A-t4);
    
            \draw[heavy edge] (A-l1) -- (A-l4);
            \draw[light edge] (A-l1) -- (A-l2);
            \draw[light edge] (A-l2) -- (A-l4);
            \draw[light edge] (A-l4) -- (A-l3);
            \draw[light edge] (A-l3) -- (A-l1);
    
            \draw[heavy edge] (A-r1) -- (A-r2);
            \draw[heavy edge] (A-r3) -- (A-r4);
            \draw[light edge] (A-r1) -- (A-r3);
            \draw[light edge] (A-r2) -- (A-r5);
            \draw[light edge] (A-r4) -- (A-r5);
            \draw[light edge] (A-r3) -- (A-r5);
            
            \draw[light edge] (A-t4) -- (A-l1);
            \draw[heavy edge] (A-t4) -- (A-r1);
            \draw[light edge] (A-t3) -- (A-r2);
            \draw[heavy edge] (A-l3) -- (A-r3);
        \end{scope}
    
        \begin{scope}[shift={(7.5,0)}]
            \node[font=\bfseries, align=center] at (0, 4.3) {Sparsifier $H$};
            \placenodes{B}
            \drawblobs
    
            \draw[discarded] (B-t2) -- (B-t3);
            \draw[light edge] (B-t1) -- (B-t2);
            \draw[light edge] (B-t1) -- (B-t3);
            \draw[light edge] (B-t2) -- (B-t4);
            \draw[light edge] (B-t3) -- (B-t4);
    
            \draw[discarded] (B-l1) -- (B-l4);
            \draw[discarded] (B-l3) -- (B-l1);
            \draw[light edge] (B-l1) -- (B-l2);
            \draw[light edge] (B-l2) -- (B-l4);
            \draw[light edge] (B-l4) -- (B-l3);
    
            \draw[discarded] (B-r1) -- (B-r2);
            \draw[discarded] (B-r3) -- (B-r4);
            \draw[light edge] (B-r1) -- (B-r3); 
            \draw[light edge] (B-r2) -- (B-r5);
            \draw[light edge] (B-r4) -- (B-r5);
            \draw[light edge] (B-r3) -- (B-r5);
    
            \draw[light edge] (B-t4) -- (B-l1);
            \draw[heavy edge] (B-t4) -- (B-r1);
            \draw[light edge] (B-t3) -- (B-r2);
            \draw[heavy edge] (B-l3) -- (B-r3);
        \end{scope}
    
        \end{tikzpicture}
        \caption{We show the \emph{balanced} case again, but this time, each cluster is an expander. The edges in $E_{\text{high}}$ are displayed in red, and the remaining edges in $E \setminus E_{\text{high}}$ are displayed in blue. Notice that some (but few) blue edges are crossing the components now.}
        \label{fig:balanced_case_expander}
\end{figure}
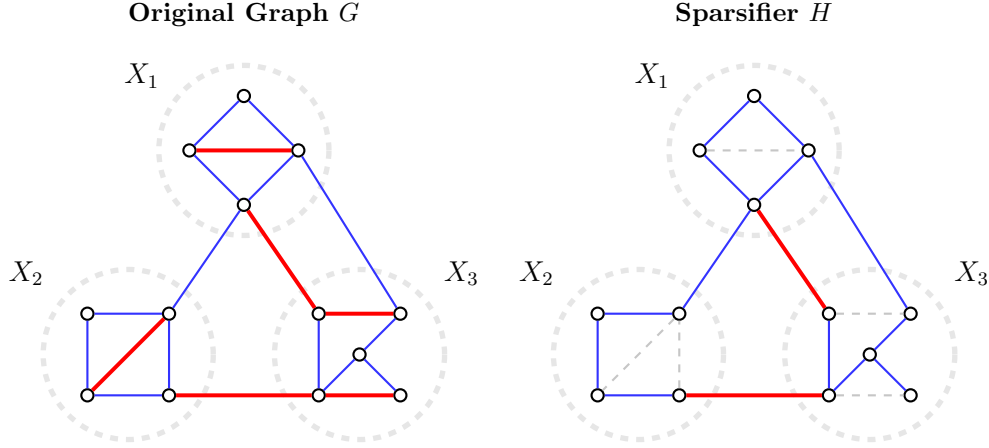

We now dynamically maintain this sparsifier $H$, while ensuring that the only edges deleted from $H$ are those deleted from $G$. Firstly, we simply add every edge that gets inserted into $G$ to $H$. Secondly, whenever an edge is deleted, we run the expander pruning algorithm on every expander this edge is in. We add the $\tilde{O}(1)$ edges that the expander pruning algorithm suggests to remove to $H$, which ensures that any MST edges among them are preserved. This ensures that we never have to start new recursive branches because the only connected components that appear have their entire volume in $H$. 
Apart from the edges added because of expander pruning, at most $\tilde{O}(1)$ edges are added due to changes in what edge connects $X_{\text{small}}$ to $X_{\text{large}}$ on various levels of the hierarchy. We remark that an edge deletion always only affects one branch of the recursive sparsifier. 

\paragraph{Further sparsifying Sparse Graphs. } If $m$ is small, the $n \log_2 n$ term from the unbalanced case becomes significant. To improve the sparsity in terms of $n$, we then contract logarithmically sized components before applying the sparsifier, which improves the sparsity to $m/16 + n - 1$. This is crucial for our algorithm because we need to obtain ultrasparsifiers to build a hierarchy.

\paragraph{A dynamic MST Algorithm via Vertex Sparsification.} Given such a powerful edge sparsifier, maintaining a minimum spanning tree becomes rather straightforward via the techniques from \cite{meierhans2026dynamic}. At a high level, we would like to just apply the sparsifier a logarithmic number of times, obtaining a sequence of graphs $G_0 = G, G_1, \ldots, G_{\Lambda}$ where $\Lambda = \tilde{O}(1)$, which directly yields the MST. Thereafter, edge deletions of $G$ cause at most $\eta = \tilde{O}(1)$ edges to be added to each sparsifier. The resulting edge insertions (or edge insertions into $G$) each lead to only one additional insertion on the next higher level, which causes only $(\Lambda+1) \cdot \eta$ changes to any level. The rather subtle point that enables us to use a logarithmic number of layers is that insertions do not cause recourse, and an edge deletion at layer $i$ causes at most one edge deletion at layer $i + 1$.

The above changes, of course, cause $G_{\Lambda}$ to no longer be an MST. To fix this, we periodically (re-)initialize the graph $G_i$ whenever $t$ is divisible by $x_i \defeq \ceil{2^{\Lambda - i - 8}/ ((\Lambda+1)\eta)}$. Notice that whenever layer $i$ is (re-)initialized, so are layers $i + 1, \ldots, \Lambda$, and that layer $\Lambda$ is (re-)initialized after every update.

Implementing the algorithm directly as discussed above would lead to prohibitively high update times because the computation of layer $\Lambda$ alone would be linear. These high-index layers, however, consist of a tree and very few additional edges. Therefore, we can apply core graphs, which are a vertex sparsification technique that essentially reduces the size of a graph to the number of off-tree edges. They work via carefully contracting paths with no adjacent off-tree edges in the tree. 

\paragraph{Precomputing in the Background for Worst-Case Update Time. }

As described, the algorithms do not run in poly-logarithmic worst-case update time. This is because occasionally, we have to rebuild a layer with a low index, which is very costly at the time. However, such rebuilds can be pre-computed, and we can therefore amortize the computation over a lot of update steps. This ultimately leads to a worst-case time algorithm for dynamic MST. 

\subsection{Dynamic Edge Sparsifier for 2-Edge Connectivity}
\label{subsec:overview_2edge_spars}

\paragraph{Dynamic $2$-Edge Connectivity Sparsifier.}

To avoid clutter, let us assume for this section that the dynamic input graph $G$ remains connected. As mentioned in \Cref{sec:intro}, a popular approach for maintaining $2$-edge connectivity of $G$ is to maintain a spanning tree $T$ and then to maintain which edges of $T$ are bridges, i.e., edges whose removal disconnects the graph. If $T$ contains no bridges, then the underlying graph $G$ is $2$-edge connected. 

While there is no certificate for being a bridge, there is one for being a non-bridge: each off-tree edge $f = xy$ forms a cycle with the $x$-to-$y$ path in the spanning tree $T$, so all edges on this path are non-bridges. We say that each such edge on the path is \emph{covered} by $f$. 

With this in mind, when we define an edge sparsifier for $2$-edge connectivity, it is crucial that we do not lose essential covering information. We thus define it as an algorithm that, given a graph $G$, maintains a subgraph $H \subseteq G$ along with a set $N \subseteq E(G)$ of non-bridge edges such that:
\begin{itemize}
    \item every edge $e \in N$ indeed has its endpoints $2$-edge connected in $G$, and 
    \item for every edge $xy \in G \setminus N$, $x,y$ are $2$-edge connected in $H$ if they are $2$-edge connected in $G$.
\end{itemize}

While the astute reader might rightfully point out that $2$-edge connectivity can be preserved straightforwardly using a sparsifier $H$ with at most $2n$ edges by using two disjoint spanning forests (see \cite{nagamochi1992linear}), the key challenge here is in obtaining ultra-sparsifier guarantees, i.e., if $m = n - 1 + k$, we still want to reduce the number of edges further to $m' = n - 1 + k/16$ (and ensuring that $H$ and $N$ change by only $\tilde{O}(1)$ edges per update to $G$).\footnote{Note that while this was also required for the MST problem, it was mitigated there by contracting small components, which is not feasible for this problem.} Note that we can assume that $k = \Omega(n)$ as we can greedily remove degree-$1$ vertices from $G$ and contract degree-$2$ vertices. The $2$-edge connectivity information of the obtained graph can be used to recover the $2$-edge connectivity information of the original graph.

\paragraph{Maintaining Coverage Information Implicitly for Expander Graphs.}

Let us denote by $T$ an arbitrary spanning tree of $G$. While the key obstacle towards worst-case algorithms is that edges in $T$ might have high congestion, i.e., be covered by many off-tree edges, and result in expensive updates, this problem can be mitigated on expanders by maintaining coverage implicitly as follows.

Consider a graph $G = (V, E)$ that is guaranteed to remain a $\phi$-expander throughout a sequence of updates. That is, for every set $\emptyset \subset C \subset V$ with $\vol_G(C) = \sum_{v \in C} \deg_G(v) \leq m$, we have that $|E(C, V \setminus C)|/\vol_G(C) \geq \phi$. Now, for any set $C$ with $2/\phi \leq \vol_G(C) \leq m$, we have at least $2$ edges leaving the set $C$, i.e, $|\partial_G C| \geq 2$. Thus, for any two sets $C_1, C_2$ of volume at least $2/\phi$, by the max-flow min-cut theorem, there are at least $2$ edge-disjoint paths starting in $C_1$ and ending in $C_2$. This insight inspires the following case distinction. For each edge $e \in T$, we say it is
\begin{itemize}
    \item \underline{volume-covered:} if both connected components of $T \setminus e$ have volumes at least $2 / \phi$. Thus, we know that $e$ is covered since, of the two edge-disjoint paths between the connected components, at least one does not contain $e$ and forms a cycle containing $e$ with $T$.
    \item \underline{cut-covered:} if the smaller side $C$ in $T \setminus e$ has less than $2/\phi$ volume, we embed the edges $f \in \partial_G C$ into $e$. Since the volume of $C$ is small, this also bounds the congestion on $e$ by $2/\phi$. We say $e$ is cut-covered if some such edge $f$ embeds into $e$, certifying that $e$ is not a bridge.
\end{itemize}
In  \Cref{fig:covered}, we illustrate the above definition. The fundamental cut induced by edge $f_1$ has a small side, and indeed $f_1$ is a bridge. The edge $f_2$, on the other hand, partitions $T$ into two large components which are well-connected in $G$ since it is an expander, and is therefore volume-covered. Finally, $f_3$ also partitions the tree into a small and large component, but this time the edge is cut-covered.

\begin{figure}[ht]
\centering
\begin{tikzpicture}[
    scale=1,
    vnode/.style={circle, draw=black, fill=white, minimum size=0.25cm, inner sep=0pt, font=\sffamily\bfseries},
    grayedge/.style={lightgray, dashed, thin},
    treeedge/.style={black, ultra thick}
]

    \def\n{10}
    \def\xradius{3cm} 
    \def\yradius{1.5cm}
    \foreach \i in {1,...,\n} {
        \pgfmathsetmacro{\angle}{360/\n*(\i-1)}
        \node[vnode] (\i) at ({\xradius * cos(\angle)}, {\yradius * sin(\angle)}) {};
    }

    \draw[grayedge] (1) -- (3); \draw[grayedge] (1) -- (4);
    \draw[grayedge] (1) -- (5); \draw[grayedge] (1) -- (6);
    \draw[grayedge] (2) -- (4); \draw[grayedge] (2) -- (5);
    \draw[grayedge] (2) -- (6);  
    \draw[grayedge] (3) -- (6); \draw[grayedge] (3) -- (8);
    \draw[grayedge] (3) -- (9); 
    \draw[grayedge] (4) -- (8); \draw[grayedge] (4) -- (10);
    \draw[grayedge] (5) -- (8); \draw[grayedge] (5) -- (9);
    \draw[grayedge] (5) -- (10); \draw[grayedge] (6) -- (9);
    \draw[grayedge] (6) -- (10);
    \draw[grayedge] (1) -- (10); \draw[grayedge] (2) -- (3);

    \draw[treeedge, red] (1) -- (2);
    \draw[treeedge] (2) -- (3);
    \draw[treeedge] (3) -- (4);
    \draw[treeedge] (3) -- (5);
    \draw[treeedge] (5) -- (6);
    \draw[treeedge, red] (6) -- (7);
    \draw[treeedge, red] (3) -- (8);
    \draw[treeedge] (8) -- (9);
    \draw[treeedge] (8) -- (10);

    \node[draw=none] at (-3,-.5) {$f_1$};
    \node[draw=none] at (-0.4,0) {$f_2$};
    \node[draw=none] at (2.9,0.7) {$f_3$};

\end{tikzpicture}
\caption{The edge $f_1$ is a bridge, $f_2$ is volume-covered, and $f_3$ is cut-covered.}
\label{fig:covered}
\end{figure}
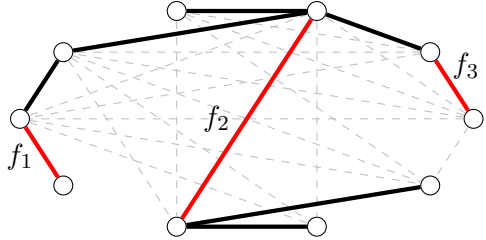

\paragraph{Our $2$-Edge Connectivity Sparsifier.}

Our algorithm first decomposes the graph into $\phi$-expanders $X_1, X_2, \ldots, X_k$ for some appropriately chosen $\phi = \tilde{\Omega}(1)$, such that less than $m/16$ edges $E_\text{rest}$ are between clusters. We then initialize for each expander $G[X_i]$ a spanning forest $T_i = \textsc{SpanningForest}(G[X_i])$ and determine for each edge in $T_i$ whether it is volume-covered, cut-covered, or a bridge. Edges in $G[X_i]$ that are not in $T_i$ are naturally non-bridges, as they form a cycle with the tree.

Then, we initialize the set $N$ to be all non-bridges of the graphs $G[X_i]$ and set $H = (G \setminus N) \cup \textsc{SpanningForest}(G \cap N)$. We have:
\begin{itemize}
    \item \underline{Correctness of $N$:} every edge $e \in N$ is not a bridge in $G$ since it is not a bridge in some $G[X_i] \subseteq G$.
    \item \ul{For every edge $e = xy$ in $G \setminus N$, $e$ is in $H$ and $2$-edge connected in $H$ if and only if it is $2$-edge connected in $G$:} To see this, observe first that for any graphs $A$ and $B$ spanning the same connected components, adding edge $e$ to both preserves this property: $A \cup e$ and $B \cup e$ have the same connected components. Since $\textsc{SpanningForest}(G \cap N)$ spans the same connected components as $G \cap N$, we have that $H \setminus e$ spans the same connected components as $G \setminus e$. But $e$ is $2$-edge connected in $H$ ($G$) if and only if its endpoints are connected in $H \setminus e$ ($G \setminus e$). The proof follows.
    \item \underline{$H$ is sparse:} Each edge in $(G \setminus N) \setminus E_\text{rest}$ is a bridge in some expander graph $G[X_i]$. Since expanders are vertex-disjoint, each such edge $e$ is also a bridge in $G \setminus E_\text{rest}$. But note that removing a bridge from a graph increases the number of connected components, and thus, the number of edges of any spanning forest by $1$. The forest $\textsc{SpanningForest}(G \cap N)$ has at most $n - |(G \setminus N) \setminus E_\text{rest}| - 1$ edges, so the number of edges in $H$ is at most
    \[
        n - |(G \setminus N) \setminus E_\text{rest}| - 1 + |G \setminus N| \leq n - 1 + |E_\text{rest}| \leq n - 1 + \frac{m}{16}.
    \]
\end{itemize}

\paragraph{Sparsifier Maintenance.}

As for the MST sparsifier, edge insertions to $G$ are straightforward; they can simply be added directly to $H$. We can thus focus on processing deletions to $G$.

First, we use the worst-case expander pruning algorithm from \cite{meierhans2025expanderpruningpolylogarithmicworstcase}, which removes $\tilde{O}(1)$ volume at each time step from some expander $G[X_i]$. We remove these edges from $N$. Each edge pruned from the expander $G[X_i]$ might result in updates to the tree $T_i$ currently spanning $G[X_i]$ and to the status of the edges in $T_i$ on whether they are volume-covered/cut-covered or bridges. While the tree $T_i$ can be repaired to be a spanning tree using our MST algorithm, our algorithm needs to do the latter tasks itself. We prove that only $\tilde{O}(1)$ edges change status, and we can efficiently find a small superset of these edges and re-evaluate their status. These proofs are rather technical, and we refer the reader to \Cref{sec:2edgeconn} for full details.

We update $N$ based on these updates, and as mentioned above, we can bound the number of changes by $\tilde{O}(1)$. Finally, we update $H$ to again equal $(G \setminus N) \cup \textsc{SpanningForest}(G \cap N)$. But note that the set $N$ has changed little, so there are also only $\tilde{O}(1)$ changes to both $G \setminus N$ and to $G \cap N$. Using our dynamic MST algorithm to process these changes, we can update $H$ by changing $\tilde{O}(1)$ edges that we find after $\tilde{O}(1)$ processing time, as desired.

%% file: preliminaries.tex
\section{Preliminaries}

\paragraph{Graphs. }

A graph $G = (V, E)$ is a tuple consisting of a vertex set $V$ and edge set $E \subseteq V \times V$. We sometimes use $n$ and $m$ to refer to the number of vertices and edges, respectively. We let $\deg_G(v)$ denote the combinatorial degree of vertex $v \in G$, i.e., the total number of edges adjacent to $v$. 

\paragraph{Cuts and Subgraphs.} We let $E(A, B) \coloneq \{uv \in E \mid u \in A \text{ and } v \in B\}$ for $A, B \subseteq V$. For notational convenience, we extend this notation to collections of vertex sets $\mathcal{X} = X_1, \ldots, X_k$ as 
\begin{align*}
    E(\mathcal{X}) = \bigcup_{1 \leq i < j \leq k} E(X_i, X_j).
\end{align*}
Furthermore, for $A \subseteq V$, we let $G[A] = (A, E(A, A))$ denote the induced subgraph of $G = (V, E)$ on $A$.  For notational convenience, we sometimes write $G \cup E'$ and $G \setminus E'$ instead of $(V, E \cup E')$ and $(V, E \setminus E')$, respectively.

\paragraph{Connected Components and Minimum Spanning Forests.}

We call a vertex set $A \subseteq V$ connected in $G$ if $E(B, A \setminus B) \neq \emptyset$ for all $B$ such that $\emptyset \subsetneq B \subsetneq A$. If $A$ is additionally maximal, i.e., there is no vertex $v \in V \setminus A$ such that $A \cup \{ v \}$ is connected, then $A$ is a connected component of $G$.

We call a graph a tree if it is connected and contains $|V| - 1$ edges. A forest is a disjoint collection of trees. We call a graph $F$ a maximal spanning forest of another graph $G$ if it is a subgraph and contains $|V| - (\text{\#connected components})$ edges. 

For the minimum spanning tree problem, we additionally associate a weight vector $\ww \in \R^{E}_{\geq 0}$ with $G$, such that the weight of each edge $e$ is $\ww(e)$. We assume that edge weights are unique without loss of generality. Then, the minimum spanning forest is the unique maximal spanning forest with minimum total edge weight.

\paragraph{Dynamic Graphs. } 

In this article, we focus on graphs whose edge sets change over time. If there are only edge insertions or only edge deletions, we call the graph incremental or decremental, respectively. We sometimes use superscripts $G^{(t)}$ to refer to the graph $G$ after the $t$-th update. 

\paragraph{Expanders. }

For $A \subseteq V$, we let $\vol_G(A) = \sum_{a \in A} \deg_G(a)$ denote the volume of $A$. We call a graph $G = (V, E)$ a $\phi$-expander if for any $\emptyset \subsetneq S \subsetneq V$ with $\min\{\vol_G(S), \vol_G(V \setminus S)\} > 0$,
\[
    \frac{|E(S, V \setminus S)|}{\min\{\vol_G(S), \vol_G(V \setminus S)\}} \geq \phi.
\]
In particular, adding singleton vertices to an expander maintains its expanding property. 

\paragraph{Expander decompositions.}

Expanders are particularly useful for maintaining connectivity information, as they are highly connected and stable under edge deletions. Specifically, most edges in an expander are not bridges, and we only need to explicitly check whether a small number of edges are bridges. We will use the following two algorithms.

\begin{definition}[Expander decomposition]
    Given a graph $G = (V, E)$, we call a partition $\mathcal{X} = \{X_1, \ldots X_k\}$ of the vertex set a $\phi$-expander decomposition of quality $q$ if $|E(\mathcal{X})| \leq q \cdot \phi \cdot |E|$ and $G[X_i]$ is a $\phi$-expander for every $i \in [k]$. 
\end{definition}

\begin{lemma}[Expander decomposition algorithm {\cite[Theorem 1.2]{SW19}}]
    \label{lem:expander_decomp}
    Given a parameter $\phi$, we can compute a $\phi$-expander decomposition of quality $q = O(\log^3 m)$ in $O\left(\frac{m \log^4 m}{\phi}\right)$ time. The algorithm succeeds with high probability.
\end{lemma}

The following lemma is obtained from Theorem 1.3 in \cite{meierhans2025expanderpruningpolylogarithmicworstcase} under the graph reduction proposed in Appendix A of \cite{chen23simple}.

\begin{lemma}[Expander pruning]
    \label{lem:expander_pruning}
    Let $G = (V, E)$ be a $\phi$-expander undergoing at most $\frac{\phi m}{\gamma}$ edge deletions where $\gamma = \tilde{O}(1)$. We can maintain a set $P \subseteq E$ such that
    \begin{itemize}
        \item $G \setminus P$ is a $\frac{\phi}{\gamma}$-expander;
        \item each deletion can be processed in $\frac{\gamma}{\phi^2}$ time;
        \item each deletion increases $P$ by at most $\frac{\gamma}{\phi^2}$ edges.
    \end{itemize}
\end{lemma}

%% file: MST.tex
\section{An Edge-Sparsifier for Fully-Dynamic Minimum Spanning Forest}
\label{sec:mst}

In this section, we introduce a novel edge sparsifier for the MST problem, which forms the central technical contribution of our work on fully dynamic minimum spanning forests. We state the result formally in the theorem below. In \Cref{subsec:completeMST}, we then show how to leverage vertex sparsification to obtain a fully dynamic minimum spanning forest algorithm.

\begin{restatable}{theorem}{edgeSparsifierMST}\label{thm:edgeSparsifierMST}
Let $G$ be a decremental graph. Then, there is an algorithm that maintains a sparsifier $H$ such that
\begin{enumerate}
    \item at any stage $t$, $\msf(H^{(t)}) = \msf(G^{(t)})$ and
    \item $|E(H^{(0)})| \leq n - \left(\text{\#connected components of }G^{(0)}\right) + m/16$ and after each deletion to $G$, at most $\tilde{O}(1)$ new edges are added to $H$.
\end{enumerate}
The algorithm takes $\tilde{O}(m)$ time to initialize. Every update to $G$ can be processed in $\tilde{O}(1)$ time. The algorithm succeeds with high probability, and explicitly reports failures. 
\end{restatable}

We first describe our core sparsification algorithm, which falls short of the sparsity guarantees in terms of the number of vertices $n$ stated in \Cref{thm:edgeSparsifierMST}. Finally, we retrieve the full theorem via a simple reduction in \Cref{sec:proof_msf_sparsifier}. 

\subsection{Constructing the Initial Sparsifier}
\label{subsec:msf_sparsifier_construction}

The construction of the sparsifier is recursive. We therefore use $|V|$ and $|E|$ when referring to the number of vertices and edges, respectively, and reserve the symbols $n$ and $m$ for the input graph. For simplicity, we slightly overload notation and refer to the current graph as $G$, but we remark that it may be a subgraph of the top-level graph. The base case is a graph with a single vertex, which has a trivial sparsifier.  

For simplicity, we assume that the expander decomposition algorithm never fails. We discuss reporting failures at the very end of this section. 

\paragraph{Algorithm Description.}

Let $E_\text{high}$ be the set of the $\frac{|E|}{64 \log_2 n}$ highest weight edges in $E$, and $E_\text{low} = E \setminus E_\text{high}$. We assume that edge weights are unique without loss of generality.

Then, we let $\phi = \frac{c}{\log_2^{5} n}$ where $c > 0$ is a sufficiently small positive constant. We first compute a $\phi$-expander decomposition of $G[E_\text{low}]$ using \Cref{lem:expander_decomp}, where we add a self-loop at each endpoint of every edge in $E_{\text{high}}$.\footnote{These self-loops ensure that for each expander, the number of edges in $E_{\text{low}}$ leaving a vertex is at least $\phi$ times the number of adjacent edges in $E_{\text{high}}$, or most of the volume of the expander is adjacent to this vertex.}
Let $\mathcal{X} = \{X_1, \dots, X_x\}$ be the expander components, and let $X_{\text{large}}$ be the largest expander in terms of vertex count, i.e., $X_{\text{large}} = X_j$ for some $j$ such that $|X_j| \geq |X_i|$ for all $i \neq j$.
We now distinguish two cases based on the size of $X_{\text{large}}$.
\begin{itemize}
    \item If $|X_{\text{large}}| \leq \frac{|V|}{2}$, we say $G$ is \textbf{balanced}. We recursively compute a sparsifier $H_i$ of each expander $G[X_i] \setminus E_{\text{high}}$, then set $H$ as the graph induced by the edge set
    \[
        E(\mathcal{X}) \cup \bigcup_{1 \leq i \leq x} E(H_i).
    \]
    We remark that $E(\mathcal{X}) = E_{\text{low}}(\mathcal{X}) \cup E_{\text{high}}(\mathcal{X})$ contains all crossing edges, not just the crossing edges of the graph $G[E_{\text{low}}]$ that we pass to the expander decomposition.
    \item If $|X_{\text{large}}| > \frac{|V|}{2}$, we say $G$ is \textbf{unbalanced}. We recursively compute a sparsifier $H_\text{small}$ of the graph $G[V \setminus X_{\text{large}}]$ and a sparsifier $H_\text{large}$ of $G[X_{\text{large}}] \setminus E_{\text{high}}$. For each $v \in V \setminus X_{\text{large}}$, let $e_v \in E_\text{high}(\{v\}, X_{\text{large}})$ be the lowest weight edge (if exists) of $E_\text{high}$ between $v$ and a vertex in $X_{\text{large}}$. Finally, we let $H$ be the graph induced by the edge set
    \[
        E(H_\text{small}) \cup E(H_\text{large}) \cup E_\text{low}(\mathcal{X}) \cup \{e_v \mid v \in V \setminus X_{\text{large}}\}.
    \]
\end{itemize}

Next, we show that the computed sparsifier $H$ contains the minimum spanning forest, which follows rather directly from the construction. 

\begin{claim}
    \label{clm:sparsifier_initial}
    The sparsifier $H$ of $G$ constructed above satisfies $\msf(H) = \msf(G)$.
    \begin{proof}
        First, we show that any $e \in E(G) \setminus E(H)$ we removed is not included in $\msf(G)$. We may remove an edge $e$ for three reasons:
        \begin{itemize}
            \item $e \in E_\text{high}$ and both its endpoints are in an expander $X$ whose edges are all in $E_\text{low}$. Since the expander is connected, $e$ must be the heaviest edge in a cycle of $X + e$, so $e \notin \msf(G)$.
            \item $G$ is in the unbalanced case, $e \in E_\text{high}(V \setminus X_{\text{large}}, X_{\text{large}})$, and there is a lighter edge $e' \in E_\text{high}(V \setminus X_{\text{large}}, X_{\text{large}})$ sharing the same endpoint of $e$ in $V \setminus X_{\text{large}}$. Write $e = uv$ and $e' = uv'$ where $u \in V \setminus X_{\text{large}}$ and $v, v' \in X_{\text{large}}$. Since $v$ and $v'$ are connected in $X_{\text{large}}$, $e$ must be the heaviest edge of a cycle in $X_{\text{large}} + \{e, e'\}$, so $e \notin \msf(G)$.
            \item $e$ is removed when we sparsify a subgraph recursively. We can conclude $e \notin \msf(G)$ for the same reasons.
        \end{itemize}
        This implies that $\msf(G) \subseteq H$ and thus $\msf(H) = \msf(G)$. 
    \end{proof}
\end{claim}

Next, we assign a level, or recursion depth, to each subgraph we sparsify for the purpose of analysing our algorithm. The top-level graph is the only level-$0$ subgraph. For each $i > 0$, the level-$i$ subgraphs are constructed from each level-$(i - 1)$ subgraph $G$. If $G$ is in the balanced case, we set each $X_i$ as a level-$i$ subgraph, and if $G$ is in the unbalanced case, we set $G[V \setminus X_{\text{large}}]$ and $G[X_\text{large}]$ as level-$i$ subgraphs.

These subgraphs form the node set of a recursion tree $\mathcal{T}$ naturally, where the root is the top-level graph, and the children of a level-$i$ subgraph are the level-$(i + 1)$ subgraphs it produces. We will refer to the children, parents, ancestors, and descendants of a subgraph in this sense.

\begin{claim}
    \label{clm:levels}
    There are at most $256 \cdot \log_2^2 n$ recursion levels.
\end{claim}
\begin{proof}
        Consider any root-to-leaf path in the recursion tree. Each vertex in this path corresponds to a subgraph. Now, let $G_i$ be a level-$i$ graph with $n'$ vertices and $m'$ edges. If $G_i$ is balanced, each expander contains at most $\frac{n'}{2}$ vertices. If it is unbalanced, then $G[V \setminus X_{\text{large}}]$ has fewer than $\frac{n'}{2}$ vertices and $G[X_{\text{large}}]$ has at most $m' - \frac{m'}{64 \cdot \log_2 n}$ edges. Since there are $n$ vertices and $m$ edges initially, the vertex decreasing can happen at most $\log_2 n$ times. Until the number of edges in the graph has halved, each iteration of edge decreasing takes away at least a $1/(128 \log_2 n)$ fraction of the initial number of edges. Therefore, the number of such iterations is bounded by $128 \log_2^2 n$. The claim follows. 
\end{proof}

\begin{claim}
    \label{clm:initial_size}
    Initially, $H$ contains at most $\frac{m}{32} + n \log_2 n$ edges.
    \begin{proof}
        Notice that every edge $e \in H$ belongs to the inter-component edges (i.e., $E(\mathcal{X})$ in the balanced case and $E_\text{low}(\mathcal{X}) \cup \{e_v \mid v \in V \setminus X_{\text{large}}\}$ in the unbalanced case) of a subgraph, so we just count those edges over all the subgraphs considered by our recursion.

        First, we bound the edges in $E_\text{low}(\mathcal{X})$ of each such subgraph. By \Cref{lem:expander_decomp}, we have
        \[
            |E_\text{low}(\mathcal{X})| = O\left(\log^3 |E|\right) \cdot \phi \cdot 2|E| < \frac{|E|}{64 \cdot 256 \log_2^2 n}.
        \] 
        Notice that the levelled graphs of the same level are vertex-disjoint, and there are at most $256 \log_2^2 n$ levels by \Cref{clm:levels}. Thus, there are at most
        \[
            256 \log_2^2 n \cdot \frac{m}{64 \cdot 256\log_2^2 n} = \frac{m}{64}
        \]
        such edges in total.

        Next, we bound the remaining edges, which are in $E_\text{high}(\mathcal{X})$ of a levelled graph. In the balanced case, all edges in $E_\text{high}(\mathcal{X})$ are included in $H$, so there are $|E_\text{high}(\mathcal{X})| \leq \frac{|E|}{64 \log_2 n}$ such edges. We partition the edge set in our recursion tree, and there are at most $\log_2 n$ balanced cases on each root-to-leaf path. Thus, the total number of remaining edges in the balanced levelled graphs is at most
        \[
            \log_2 n \cdot \frac{m}{64 \log_2 n} = \frac{m}{64}.
        \]
        In the unbalanced case, we only preserve the edge $e_v$ for each $v \in V \setminus X_{\text{large}}$, so there are at most $|V \setminus X_{\text{large}}| < \frac{|V|}{2}$ such edges. We charge each such edge to its endpoint in $V \setminus X_{\text{large}}$. Since the component of a vertex can only be half in size at most $\log_2 n$ times, every vertex can be charged at most $\log_2 n $ times. Therefore, the total number of such edges is at most $n \log_2 n$.
        
        We conclude that there are at most
        \[
            \frac{m}{64} + \frac{m}{64} + n\log_2 n = \frac{m}{32} + n \log_2 n
        \]
        edges in $H$ in total.
    \end{proof}
\end{claim}

\subsubsection{Processing Edge Deletions}

In this section, we show how to maintain the construction from \Cref{subsec:msf_sparsifier_construction} under edge deletions. The main technical tool enabling the maintenance is worst-case expander pruning, which allows us to maintain an expander throughout a sequence of edge deletions. 

\paragraph{Additional Initialization}

As above, we first do an expander decomposition and obtain the vertex clusters $X_1, \dots, X_x$. For each expander $X_i$, we set up a pruning data structure $\mathcal{P}_{X_i}$ using \Cref{lem:expander_pruning}. The data structure maintains an initially empty edge set $P \subseteq E(X_i)$, such that $(G \setminus P)[X_i \setminus S]$ remains an expander, where $S$ contains all vertices in $(G \setminus P)$ with volume $0$, i.e., isolated vertices. Note that according to \Cref{lem:expander_pruning}, every $\mathcal{P}_{X_i}$ has a deletion limit.

If $G$ is unbalanced, we initialize a data structure $\mathcal{D}_G$ consisting of $|V \setminus X_{\text{large}}|$ min-heaps for each $v \in V \setminus X_{\text{large}}$. Each $\mathcal{D}_G(v)$ stores the edges of $E_\text{high}(\{v\}, X_{\text{large}})$ according to their weights. 

These data structures are maintained for every level of the recursion. 

\paragraph{Handling edge deletions}

Suppose we delete edge $e$ from $G$. We delete $e$ from $H$ (if $e \in H$) and update $H$ recursively. Below, we describe how the update of $H$ is processed for the top-level graph $G$. We will perform these operations for each recursion level.

First, if $G$ is unbalanced, we remove $e = uv$ from the queue $\mathcal{D}_G(w)$ that contains $e$ where $w \in \{u, v\}$. Notice that only edges with one endpoint in $X_{\text{large}}$ end up in a priority queue, and therefore at most one of the two can contain $e$. If $e$ was the lightest edge connecting one of the vertices in $X_{\text{small}} = V \setminus X_{\text{large}}$ to $X_{\text{large}}$, we replace it with the lowest-weight edge in the corresponding priority queue. We add this edge to $H$.

Regardless of whether $G$ is balanced or unbalanced, we then check if $e$ is an edge internal to an expander $G[X_i]$. If this is the case, we inform $\mathcal{P}_{X_i}$ of the deletion of $e$. If $\mathcal{P}_{X_i}$ has not reached the deletion limit, we obtain a set of edges $\Delta P \subseteq X_i$ whose removal preserves the expander. Add all these edges in $\Delta P$ to $H$, and additionally add $\frac{2\gamma}{\phi}$ edges incident to $X_i$ (if they exist) into $H$ where $\gamma = \tilde{O}(1)$ is as given in \Cref{lem:expander_pruning}. If $\mathcal{P}_{X_i}$ reaches its deletion limit, we claim below that by the extra additions, every edge incident to $X_i$ is already included in $H$, so we no longer need to maintain $\mathcal{P}_{X_i}$. We say $X_i$ is \textbf{active} if $\mathcal{P}_{X_i}$ has not reached the limit and $X_i$ is \textbf{inactive} otherwise.

In the unbalanced case, we further add some of the edges in $E_{\text{high}}(X_{\text{small}}, X_{\text{large}})$ to $H$ and remove them from the priority queues. For every vertex $v$ that is incident to either the deleted edge, one of the edges output by expander pruning, or one of the arbitrarily selected additional edges, we remove $1/\phi$ arbitrary additional edges incident to $v$ (if they exist) of $E_{\text{high}}(X_{\text{small}}, X_{\text{large}})$ from their respective heaps and add them to $H$. If this causes the minimum maintained by the heaps to change, we additionally add the new minimum-weight edges. 

\begin{claim}
    \label{clm:singleton_expander}
    Consider an expander component $X_i$, and let $P$ be the set maintained by $\mathcal{P}_{X_i}$. If a vertex $x \in X_i$ becomes a singleton in $G[X_i] \setminus P$, then all remaining edges incident to $x$ are in $H$.
\end{claim}

\begin{proof}
    Since $x$ becomes a singleton, every edge of $G[X_i] \setminus E_{\text{high}}$ incident to $x$ has either already been deleted from $G$ or placed into the pruned set $P$. Whenever an edge is added to $P$, our algorithm immediately adds it to $H$. Therefore, all remaining edges of $G[X_i] \setminus E_{\text{high}}$ incident to $x$ are in $H$.

    It remains to account for the edges in $E_{\text{high}}$ incident to $x$. Since we additionally add $1 / \phi$ edges adjacent to every endpoint per deleted or pruned edge to $H$, and there are at most $1/\phi$ times the initial degree of $x$ such edges by expansion (because we added a self-loop for each such edge before computing the expander decomposition), these are also all in $H$. This concludes the proof.  
\end{proof}

\begin{claim}
    \label{clm:inactive_expander}
    If an expander $X$ is inactive, every edge incident to $X$ is added to $H$.
    \begin{proof}
        Suppose $X$ has $m'$ edges initially. When $\mathcal{P}_X$ reaches its deletion limit, $X$ underwent $\frac{\phi m'}{\gamma}$ edge deletions. Since $\frac{2\gamma}{\phi}$ edges incident to $X$ are added to $H$ for each deletion, we have already added $\frac{\phi m'}{\gamma} \cdot \frac{2\gamma}{\phi} = 2m'$ edges to $H$, so all edges incident to $X$ are in $H$. By expansion, this also holds for the edges in $E_{\text{high}}$ adjacent to $X$ because every vertex of degree $D$ in $G[X] \setminus E_{\text{high}}$ is incident to at most $D/\phi$ edges in $E_{\text{high}}$ (since we add self-loops for them in our algorithm). These get included in $H$ by the description of our algorithm.
    \end{proof}
\end{claim}

\begin{claim}
\label{clm:point_to_remainder}
    Assume that we are in the unbalanced case. Then, either all edges in data structures $\mathcal{D}(\cdot)$ are adjacent to the same connected component of $G[X_{\textnormal{large}}] \setminus P$, or all edges adjacent to $X_{\textnormal{large}}$ are in $H$.
\end{claim}
\begin{proof}
     We can assume that we are in the active case by \Cref{clm:inactive_expander}. Then, every time an edge gets pruned or deleted, we remove $1/\phi$ edges that are adjacent to the endpoints from the heaps. Again, by expansion and since we add a self-loop for each such edge, all of them are removed once a vertex becomes a singleton. By expansion, every non-singleton vertex is connected.
\end{proof}

\begin{claim}
    \label{clm:msf_maintaining}
    After deleting any edge $e \in E$ and updating $H$, we still have $H \subseteq G$ and $\msf(H) = \msf(G)$. Moreover, $H$ increases by at most $\tilde{O}(1)$ edges.

    \begin{proof}
        We must have $H \subseteq G$ since we delete $e$ from $H$.
        
        Let $f \in E \setminus \{e\}$ be an edge not in $H$. To show $\msf(H) = \msf(G)$, it suffices to show that $f \notin \msf(G)$. By our algorithm, \Cref{clm:singleton_expander}, and \Cref{clm:inactive_expander}, $f$ cannot be incident to a singleton vertex of the maintained expander or in an inactive expander. Thus, $f$ is contained in an unpruned part of an expander or between two of them. If it is between them, it has to be heavy, so applying \Cref{clm:point_to_remainder} and following the proof of \Cref{clm:sparsifier_initial} gives $f \notin \msf(G)$.
        We add at most
        \[
            O\left(1 + |\Delta P| + \frac{2\gamma}{\phi} + \frac{1}{\phi}\left(1 + |\Delta P| + \frac{2\gamma}{\phi}\right)\right) \leq O\left(\frac{1}{\phi} \cdot \left(1 + \frac{\gamma}{\phi^2} + \frac{2\gamma}{\phi}\right)\right) = \tilde{O}(1)
        \]
        edges to $H$ in each recursion level. Since there are $\tilde{O}(1)$ levels, we add $\tilde{O}(1)$ edges to $H$.
    \end{proof}
\end{claim}

\subsection{Bounding the Required Time for Initialization and Updates}

\begin{claim}
    \label{clm:runtime_update}
    After each edge deletion, we need $\tilde{O}(1)$ time to update $H$.
    \begin{proof}
        By \Cref{clm:levels}, at most $\tilde{O}(1)$ many heaps contain $e$. Removing $e$ from each heap costs $O(\log m)$ time. Updating each $\mathcal{P}_{X_i}$ and finding the pruned edges $\Delta P$ can be done in $\frac{\gamma}{\phi^2} = \tilde{O}(1)$ time. The additional pruning and maintenance of the heaps can be done in $\tilde{O}(1/\phi^3)$ time. Finding the next lightest edge $e'$ costs $O(\log m)$ time. All other operations can be done in $\tilde{O}(1)$ time.
    \end{proof}
\end{claim}

\begin{claim}
    \label{clm:runtime}
    We need time $\tilde{O}(m)$ to construct $H$ and initialize each $\mathcal{D}_G$ and $\mathcal{P}_{X_i}$.
    \begin{proof}
        Fix a level $i$ and let $G_1, \dots, G_s$ be the level-$i$ subgraphs (note that they are disjoint). Let each $G_j$ contain $n_j$ vertices and $m_j$ edges. The time we spend on this level includes
        \begin{itemize}
            \item Computing a $\phi$-expander decomposition of each $G_j$. By \Cref{lem:expander_decomp}, this cost time
            \[
                \sum_{1 \leq j \leq s} \tilde{O}\left(\frac{m_j}{\phi}\right) \leq \tilde{O}\left(\frac{m}{\phi}\right) = \tilde{O}(m).
            \]
            \item Removing edges of $E_\text{high}$ in each expander, which can be done in $O(m)$ time by a linear scan.
            \item (For each unbalanced $G_j$) Removing edges of $E_\text{high}(\mathcal{X}) \setminus \{e_v \mid v \in V \setminus X_{\text{large}}\}$. This can be done in $O(m \log m) = \tilde{O}(m)$ time by sorting the edges.
            \item Constructing $\mathcal{D}_{G_j}$ (for each unbalanced $G_j$) and $\mathcal{P}_{X_i}$. They can be done in $\tilde{O}(m)$ time.
        \end{itemize}
        Since there are $\tilde{O}(1)$ levels by \Cref{clm:levels}, the total runtime is $\tilde{O}(m)$.
    \end{proof}
\end{claim}

\subsection{Proof of \texorpdfstring{\Cref{thm:edgeSparsifierMST}}{Main Theorem}}
\label{sec:proof_msf_sparsifier}

We finally assemble the results shown above to prove \Cref{thm:edgeSparsifierMST}, the main theorem of this section. We assume that the graph is initially connected; otherwise, we run the algorithm separately on each component. By \Cref{clm:initial_size}, we initially reduce the edge count of a graph with $m$ edges and $n$ vertices to $m/32 + n \log_2 n$. If $n \log_2 n \leq m/32$, \Cref{thm:edgeSparsifierMST} directly follows from \Cref{clm:sparsifier_initial}, \Cref{clm:msf_maintaining}, and \Cref{clm:runtime}.

We now describe how we handle the very sparse case when $n \log_2 n > m/32$, i.e., $m < 32 \cdot n \log_2 n$. We replace every vertex $v$ with degree $D > 3$ with a collection of $D$ vertices and assign one of the edges of $v$ to each such vertex. Then, we add a path between these vertices, and give all the edges on the path the weight $-\infty$. Clearly, the MST of the graph remains the same after contracting all edges with weight $-\infty$. The new graph has $n' \leq 3m$ vertices and $m' \leq 3m$ edges.

We will now introduce the concept of shattering forests, which allows us to decompose this minimum spanning tree into small chunks. This will then let us work on a graph with fewer vertices, achieving the claimed density bound in this very sparse case.  

\begin{definition}[Shattering Forest] \label{def:shattering_forest}
Given graphs $G = (V, E)$ and $H \subseteq G$, we call a forest $F \subseteq G$ composed of trees $T_1, \ldots, T_\eta$ a $\kappa$-shattering forest with respect to $H$ if $\vol_H(V(T_i)) \leq \kappa$ for all $i \in 1, \ldots, \eta$ and $V(G) = V(F)$. 
\end{definition}

\begin{lemma}[see \cite{frederickson1985data}, Lemma 1]\label{lma:clusteringFrederikson}
For an $n$-vertex connected tree $T$ with maximum degree $3$, and a positive integer $z$ such that $n \geq 3z-2$, there is an algorithm that runs in time $O(n)$ and returns a set $E' \subseteq E(T)$ such that each connected component in $T \setminus E'$ contains at least $z$ and at most $3z-2$ vertices.
\end{lemma}

The minimum spanning tree in our graph now has maximum degree $3$, and we can split it into components of size at least $z = 3 \cdot 128 \log_2^2 n$ by \Cref{lma:clusteringFrederikson}. There are $\eta \leq \frac{n'}{z} \leq \frac{m}{128 \log_2^2 n}$ such components, each having a volume at most $9z = \tilde{O}(1)$. We contract each such component, and then run the algorithm described above on this graph. Note that this graph has at most $m' - (n' - \eta) = m' - n' + \eta = m - n + \eta \leq m$ edges. The initial sparsifier on this contracted graph then has at most
\[
    \frac{m}{32} + \frac{n'}{z} \log_2 \frac{n'}{z} \leq \frac{m}{32} + \frac{m}{128 \log_2^2 n} \cdot 4\log_2 n \leq \frac{m}{16}
\]
edges, and uncontracting adds at most $n - 1$ additional edges. 

When returning the sparsifier, we contract the $-\infty$-edges that correspond to a single vertex, and whenever an update touches a component (i.e., an endpoint of a deleted edge lies in a component), we add all edges incident to that component to the final sparsifier, which only adds $\tilde{O}(1)$ edges. Correctness again follows directly from \Cref{clm:sparsifier_initial}, \Cref{clm:msf_maintaining}, and \Cref{clm:runtime}.

We finally address potential failures in the expander decomposition algorithm, since this algorithm only succeeds with high probability. If the expander decomposition is deficient, one of the components is not a $\phi$-expander, or there are too many crossing edges. If there are too many crossing edges, we detect the failure immediately. Otherwise, a component $X_i$ of the expander decomposition could become disconnected, which is the only reason our algorithm would fail to return a proper MST sparsifier. We can avoid this by explicitly maintaining the connectivity of the components with the algorithm of \cite{meierhans2026dynamic}. If the expander decomposition gets derandomized in the future, this step will no longer be necessary. 

This concludes the proof of \Cref{thm:edgeSparsifierMST}.

%% file: 2EdgeConnectivity.tex
\section{An Edge-Sparsifier for Fully-Dynamic 2-Edge Connectivity}
\label{sec:2edgeconn}

In this section, we introduce a novel edge sparsifier for the $2$-edge connectivity problem. This is the key primitive required to solve the fully-dynamic $2$-edge connectivity problem, as we show in \Cref{subsec:complete2EdgeConn}.

\begin{restatable}{theorem}{edgeSparsifier2Conn}
\label{thm:edgeSparsifier2conn}
Let $G$ be a decremental graph. Then, there is an algorithm that maintains a sparsifier $H$ along with a set of $2$-edge connected edges/non-bridges $N \subseteq E(G)$ such that: 
\begin{enumerate}
    \item at any stage $t$, $H^{(t)} \subseteq G^{(t)}$ and $H^{(t)}$ always contains a maximal spanning forest of $G^{(t)}$, and
    \item $|E(H^{(0)})| \leq n - (\text{\#connected components of }G^{(0)}) + m/16$ and after each deletion to $G$, at most $\tilde{O}(1)$ new edges are added to $H$, and
    \item the set $N \subseteq E(G^{(t)})$ is decremental, has at most $\tilde{O}(1)$ edges removed after each edge deletion to $G$, and for each edge $e \in E(G^{(t)})$, we have:
    \begin{itemize}
        \item \underline{if $e \in N^{(t)}$:} then the endpoints of $e$ are $2$-edge connected in $G^{(t)}$, and
        \item \underline{otherwise:} the endpoints of $e$ are $2$-edge connected in $G^{(t)}$ if and only if they are $2$-edge connected in $H^{(t)}$.
    \end{itemize}
\end{enumerate}
The algorithm takes $\tilde{O}(m)$ time to initialize. Every update to $G$ can be processed in $\tilde{O}(1)$ time. The algorithm succeeds with high probability. 
\end{restatable}

\subsection{The Algorithm}
\label{subsubsec:decremental_expander_algorithm}

\paragraph{Decomposing the Graph into Expanders.}
We use the algorithm from  \Cref{lem:expander_decomp} to compute a $\phi$-expander decomposition $\mathcal{X}$ of $G$ for $\phi = 1/(16q) = \tilde{\Omega}(1)$, which ensures that $\partial_G \mathcal{X}$ contains at most $m/16$ edges of $G$ between clusters. 

We henceforth refer to each expander $X \in \mathcal{X}$ as \emph{large} if $G^{(0)}[X]$ has volume at least $4\gamma/\phi$, and otherwise we say it is \emph{small}. We let $\mathcal{X}^\text{large}$ denote the set of large expanders and $\mathcal{X}^\text{small}$ the set of small expanders. For each $X \in \mathcal{X}^\text{large}$, we initialize a data structure $\mathcal{P}_X$ from \Cref{lem:expander_pruning} on $G[X]$ and let $P_X^{(t)}$ be the set of all edges returned by the data structure after processing the first $t$ deletions to $G$. Each deletion in $G$ only affects at most one induced graph $G[X]$, so $P_Y$ is not affected for all $Y \neq X$ at that time step. 

For $X \in \mathcal{X}^\text{large}$, we define $\bar{G}_X^{(t)} = G^{(t)}[X] \setminus P_X^{(t)}$. Recall that $\bar{G}_X^{(t)}$ remains a $\phi/\gamma$-expander ($\gamma = \tilde{O}(1)$, see \Cref{lem:expander_pruning}). For small expanders $X \in \mathcal{X}^\text{small}$, we simply set $\bar{G}_X^{(t)} = G^{(t)}[X]$. 

For each expander $X \in \mathcal{X}$, we maintain a local maximal spanning tree $\bar{F}_X^{(t)}$ of $\bar{G}_X^{(t)}$. To uniquely define $\bar{F}_X^{(t)}$, we assign edges $e \in \bar{G}_X$ unique weights at initialization, and maintain $\bar{F}_X^{(t)}$ to be $\msf(\bar{G}_X^{(t)})$. Note that an edge in $\bar{F}_X$ is only removed once the edge is deleted by $\bar{G}_X^{(t)}$. Any edge that becomes a bridge, i.e., not $2$-edge connected, will be part of $\bar{F}_X^{(t)}$ at this stage, because it is the only edge connecting the connected components of its endpoints.

\paragraph{Non-Bridge Sets via Covers.}

The main part of our algorithm is to maintain the set $N$ of edges whose endpoints are $2$-edge connected in their expander. More precisely, for each expander $X \in \mathcal{X}$ we let $N_X^{(t)}$ be the set of edges in $\bar{G}^{(t)}_X$ that are $2$-edge connected/non-bridges in $\bar{G}^{(t)}_X$. Finally, we let  $N = \bigcup_{X \in \mathcal{X}} N_X$. Note that once an edge is no longer $2$-edge connected in its expander, it remains a bridge for the rest of the algorithm because the graph is decremental. 

For small expanders $X \in \mathcal{X}^\text{small}$, after any edge deletion to $G[X]$, we can simply update the set $N_X$ to consist of all edges in $\bar{G}_X^{(t)}$ whose endpoints are $2$-edge connected in $\bar{G}_X^{(t)}$. 

Let us next discuss how we maintain the set of $2$-edge connected edges/non-bridges $N_X$ for large expanders $X \in \mathcal{X}^\text{large}$. We need the following definitions.

\begin{definition}[Covered]
    Given a graph $G'$ and a forest $F'$ of $G'$. For a vertex $x$ (and an edge $f$ in $G'$), we denote by $C_{F', x, f}$ the connected component of $F' \setminus f$ that contains $x$. We say an edge $f \in F'$ is
    \begin{itemize}
        \item \emph{volume-covered} if for its endpoints $x,y$, we have $\vol_{G'}(C_{F', x, f}), \vol_{G'}(C_{F', y, f}) \geq 2\gamma/\phi$, and
        \item \emph{cut-covered} if for at least one endpoint $y$ of $f$, we have $\vol_{G'}(C_{F', y, f}) < 2\gamma/\phi$ \textbf{and} there is an edge $f' \in \partial_{G'} C_{F', y, f}$ with $f' \neq f$. We call such an edge $f'$ a \emph{covering edge} for $f$.
    \end{itemize}
    See \Cref{fig:covered} for an illustration. For any stage $t$, we say an edge $f \in G'$ is \emph{covered} if $f \not\in F'$ \textbf{or} $f \in F'$ and $f$ is volume- or cut-covered.
\end{definition}

Henceforth, we let $N_X^{(t)}$ be the set of all edges $f \in \bar{G}_X^{(t)}$ that are covered in $\bar{G}_X^{(t)}$ with forest $\bar{F}^{(t)}_X$. We defer implementation details on how to maintain the covering information until later.

\paragraph{The Sparsifier and Non-bridge Set.}

We output the set $N = \bigcup_{X \in \mathcal{X}} N_X$ to be the set of edges that are non-bridges. Finally, we maintain $H^{(t)} = (G^{(t)} \setminus N^{(t)}) \cup \msf(G^{(t)} \cap N^{(t)})$, where we again use the identifiers of edges as edge weights to uniquely determine the minimum spanning forest. Note that every edge $f$ in some local spanning forest $\bar{F}_X^{(t)}$ and all pruned sets $P_X^{(t)}$ are in $H^{(t)}$.

\subsection{Analysis}

For expander $X \in \mathcal{X}$, stage $t$, vertex $x \in X$, and edge $f$ in $\bar{G}^{(t)}_X$, we denote the connected component $C_{\bar{F}_X^{(t)}, x, f}$ by $C^{(t)}_{x,f}$. We further write $C^{(t)}_x$ to mean the connected component containing $x$ in $\bar{F}_X^{(t)}$.

\subsubsection{Correctness}

We first prove the following fact about expansion.

\begin{fact}\label{prop:expander_center}
    For any $\phi$-expander graph $A$ and disjoint sets $S, T \subseteq V(A)$ with 
    \[
        \vol_A(S), \vol_A(T) \geq \frac{2}{\phi},
    \]
    $S$ and $T$ are $2$-edge connected in $A$, or put differently, there are at least two edge-disjoint paths starting in $S$ and ending in $T$.
    \begin{proof}
    Let $(S^*, T^* \coloneqq V(A) \setminus S^*)$ be any cut that separates $S$ and $T$, i.e., $S \subseteq S^*, T \subseteq T^*$. If $\vol_A(S^*) \leq \vol_A(T^*)$, then
    \[
        |E_A(S^*, T^*)| \geq \phi \cdot \vol_A(S^*) \geq \phi \cdot \vol_A(S) \geq 2.
    \]
    The case where $\vol_A(S^*) > \vol_A(T^*)$ follows analogously. This proves that every $S$-$T$ cut is of size at least $2$, and therefore, by the max-flow min-cut theorem, there are at least two disjoint paths from $S$ to $T$, as desired.
    \end{proof}
\end{fact}

\Cref{prop:expander_center} can be used to establish that endpoints of edges in $N_X^{(t)}$ are $2$-edge connected. We now prove our key invariant.

\begin{restatable}{lemma}{invBridgeIff}\label{inv:nonBridgeIffCovered}
    For any stage $t$, expander $X \in \mathcal{X}$ and an edge $f \in \bar{G}_X^{(t)}$, $f$ is a non-bridge in $\bar{G}_X^{(t)}$ if and only if $f \in N_X^{(t)}$.
\end{restatable}
\begin{proof}
For $X \in \mathcal{X}^\text{small}$, the claim is true by construction, so assume $X \in \mathcal{X}^\text{large}$ below. 

For $f = xy \in \bar{G}^{(t)}_X \setminus \bar{F}^{(t)}_X$, since $f$ was not added to $\bar{F}^{(t)}_X$, $x$ and $y$ are spanned by $\bar{F}^{(t)}_X$. Thus, we can travel from $x$ to $y$ either along their unique path in $\bar{F}^{(t)}_X$, or along $f$. Thus, $x$ and $y$ are $2$-edge connected. Each such edge is added to $N_X^{(t)}$.

For $f \in \bar{F}^{(t)}_X$, we have $f \in N_X^{(t)}$ if either
\begin{itemize}
    \item \underline{$f$ is \emph{volume-covered}:} Since, by definition, $\vol_{\bar{G}^{(t)}_X}\left(C_{y,f}^{(t)}\right), \vol_{\bar{G}^{(t)}_X}\left(C_{x,f}^{(t)}\right) \geq 2\gamma/\phi$, we have by \Cref{prop:expander_center} that there are at least $2$ edge-disjoint paths in $\bar{G}_X^{(t)}$ between $C_{y,f}^{(t)}$ and $C_{x,f}^{(t)}$. Thus, there is at least one such path $P$ that does not contain $f$. But since both $C_{y,f}^{(t)}$ and $C_{x,f}^{(t)}$ are spanned by $\bar{F}_X^{(t)} \setminus f$, there is an $x$-to-$y$ path in $\left(\bar{F}_X^{(t)} \setminus f\right) \cup P$, so $x$ and $y$ are $2$-edge connected.

    \item \underline{$f$ is \emph{cut-covered}:} Then there is an edge $f' \in \partial_{\bar{G}^{(t)}_X} C^{(t)}_{x,f}$ for $x$ being an endpoint of $f$ with $f' \neq f$. But since $\bar{F}^{(t)}_X$ is a maximal spanning forest and $f'$ is not in $\bar{F}_X^{(t)}$, we have that $\left(\bar{F}^{(t)}_X \setminus f\right) \cup f'$ spans the same connected components and thus contains a path from $x$ to $y$. Thus, again, $x$ and $y$ are $2$-edge connected.
\end{itemize}

Finally, consider that $f \in \bar{F}_X^{(t)}$ is not covered and therefore not in $N^{(t)}_X$. Then $f$ is not volume-covered and at least one endpoint $x$ of $f$ has $\vol_{\bar{G}^{(t)}_X}(C_{x,f}^{(t)}) < 2\gamma/\phi$. But since $f$ is also not cut-covered, it must have $\partial_{\bar{G}^{(t)}_X} C^{(t)}_{x,f} = \{f\}$, i.e., $f$ is the only edge in the cut, implying that $f$ is a bridge.
\end{proof}

Since $N^{(t)} = \bigcup_{X \in \mathcal{X}} N_X^{(t)}$ and $\bar{G}_X^{(t)} \subseteq G^{(t)}$ for every $X \in \mathcal{X}$, the following corollary is immediate.

\begin{corollary}
For any time $t$ and edge $f \in N^{(t)}$, the endpoints of $f$ are $2$-edge connected in $G^{(t)}$.
\end{corollary}

It remains to establish that it can be recovered from $H^{(t)}$ whether an edge $e \in G^{(t)} \setminus N^{(t)}$ is a bridge or $2$-edge connected. 

\begin{lemma}
If $e \in G^{(t)} \setminus N^{(t)}$, then the endpoints of $e$ are $2$-edge connected in $G^{(t)}$ if and only if they are $2$-edge connected in $H^{(t)}$.
\end{lemma}
\begin{proof}
We first recall the following standard graph-theoretic fact.

\begin{fact}
For graphs $A$ and $B \subseteq A$, suppose the connected components of $A$ and $B$ coincide. Then, for any edge $e$, the connected components of graphs $A \cup e$ and $B \cup e$ coincide. 
\end{fact}

Now, since $\msf(G^{(t)} \cap N^{(t)})$ spans the connected components of $G^{(t)} \cap N^{(t)}$ (with minimum weight edges), we get that $H^{(t)} \setminus e = \msf(G^{(t)} \cap N^{(t)}) \cup (G^{(t)} \setminus (N^{(t)} \cup e))$ spans the connected components of 
$(G^{(t)} \cap N^{(t)}) \cup (G^{(t)} \setminus (N^{(t)} \cup e)) = G^{(t)} \setminus e$. For $e = uv$, $u$ and $v$ are connected in $G^{(t)} \setminus e$ if and only if they are connected in $H^{(t)} \setminus e$. The proof follows.
\end{proof}

\subsubsection{Recourse Analysis}
\label{subsubsec:recourse}

It remains to argue about the recourse. To enable a rigorous analysis, we define and prove various properties of forests undergoing changes. Here, we talk about a general vertex measure $\mu$, but the reader can imagine this measure referring to the volume of the underlying graph $G$.  

\begin{definition}[Centroid]
    Assume we have a forest $F' = (V', E')$ and a measure $\mu \colon V' \mapsto \mathbb{N}_{\geq 0}$. Let $C_{F', x}$ be the connected component in $F'$ containing vertex $x$. Then, we let $c_{F', \mu}(C_{F',x})$ denote the centroid of $C = C_{F',x}$, that is, the vertex with the largest identifier such that each connected component of $F'[C] \setminus \{c\}$ contains at most $\mu(C) / 2$ mass from measure $\mu$. 

    We use the shorthand $c_{F', \mu}(x) = c_{F', \mu}(C_{F',x})$, for convenience. For $x \neq c_{F', \mu}(x)$, we denote by $T_{F', \mu, x}$ the subtree containing $x$ obtained from removing the first edge on the $x$-to-$c_{F', \mu}(x)$ path from $F'[C_{F', x}]$. Note that $\mu(T_{F', \mu, x}) \leq \mu(C_{F', x}) / 2$ by definition for each such $x \neq c_{F', \mu}(x)$. When $F'$ and $\mu$ are clear from context, we often omit them in the subscript.
\end{definition}

We use the following two properties of centroids, whose proofs can be found in \Cref{app:centroid}.

\begin{restatable}[Small Subtrees are Robust to Measure Changes]{claim}{robustToMeasureChange}
    \label{clm:robustToMeasureChange}
    Suppose we have a forest $F'$ and a measure $\mu$ such that $\mu(v) \geq \deg_{F'}(v)$. Let $\mu'$ be a measure obtained by increasing or decreasing a single coordinate of $\mu$ by $1$. Then, the set $Q$ of all vertices $q$ in $V'$ where $\mu'(T_{F', \mu', q}) \neq \mu(T_{F', \mu, q})$ and $\mu'(T_{F', \mu', q}) < 2\gamma/\phi$ is of size $\tilde{O}(1)$.
\end{restatable}

\begin{restatable}[Small Subtrees are Robust to Forest Changes]{claim}{robustToForestChange}
    \label{clm:robustToForestChange}
    Consider any forest $F'$ and measure $\mu$ such that $\mu(v) \geq \deg_{F'}(v)$. Consider the update of inserting or deleting an edge $f$ into/from $F'$. Let $F''$ be the resulting forest. Then, the set $Q$ of all vertices $q$ in $V'$ where $\mu(T_{F', \mu, q}) \neq \mu(T_{F'', \mu, q})$ and $\mu(T_{F'', \mu, q}) < 2\gamma/\phi$ is of size $\tilde{O}(1)$.
\end{restatable}

We can now argue about the number of edges that need to be checked for whether they are covered after each update to $\bar{G}_X$. 

\begin{claim}
    \label{clm:littleVolumeInSmallComponents}
    For any $X \in \mathcal{X}$, any edge deletion in $\bar{G}_X$ has at most $\tilde{O}(1)$ edges changing their covering status. That is, they transition between being covered (volume- or cut-covered) and uncovered, or they enter/leave $\bar{F}_X$.
\end{claim}

\begin{proof}
    Let $e = uv$ be the edge deleted from $\bar{G}_X$ and $\bar{G}_X' = \bar{G}_X \setminus \{e\}$. Let $\bar{F}'_X$ denote the updated minimum spanning forest in $\bar{G}_X'$, which is obtained from $\bar{F}_X$ by removing $e$ (if $e \in \bar{F}_X$) and possibly inserting a replacement edge $e'$. We henceforth use $c(x)$ as shorthand for $c_{\bar{F}_X, \vol_{\bar{G}_X}}(x)$ and $c'(x)$ for $c_{\bar{F}'_X, \vol_{\bar{G}'_X}}(x)$. Similarly, we use $T_x$ as shorthand for $T_{\bar{F}_X, \vol_{\bar{G}_X}, x}$ and define $T'_x$ with respect to $\bar{F}_X'$ and $\bar{G}_X'$.

    Deleting $e$ results in it being removed and possibly replaced by an edge $e'$ in $\bar{F}_X$. We can emulate this process by at most one deletion from $\bar{F}_X$, one insertion into $\bar{F}_X$, and two decreases in the volume measure $\vol_{\bar{G}_X}$ at the endpoints. Thus, we have from \Cref{clm:robustToMeasureChange} and \Cref{clm:robustToForestChange}, that the number of vertices $q \in Q$ with $T'_q \neq T_q$\footnote{The monotonicity of single updates ensures that $T'_{q} \neq T_q$ results in $\vol_{\bar{G}'_X}(T'_q) \neq \vol_{\bar{G}_X}(T_{q})$, and so this is implied in \Cref{clm:robustToMeasureChange} and \Cref{clm:robustToForestChange}, but here we need to be slightly more careful.} or $\vol_{\bar{G}'_X}(T'_{q}) \neq \vol_{\bar{G}_X}(T_{q})$ \textbf{and} $\vol_{\bar{G}'_X}(T'_q) < 2\gamma/\phi$ is of size $\tilde{O}(1)$.

    We argue by cases:
    \begin{itemize}
        \item \underline{Edge $e$ and replacement edge $e'$:} we count them as edges that change their status; however, they contribute $O(1)$ edges. All edges $f$ considered below are taken to be $f \neq e, e'$.
        \item \underline{Change in volume-covered status:} Let an edge $f$ be volume-covered in $\bar{G}_X$ but not in $\bar{G}_X'$. Then, at least one endpoint $q$ has
        \[
            \vol_{\bar{G}_X'}\left(C_{\bar{F}_X', q, f}\right) < \frac{2\gamma}{\phi} \leq \vol_{\bar{G}_X}\left(C_{\bar{F}_X, q, f}\right).
        \]
        We may assume that $\vol_{\bar{G}'_X}\left(C_{\bar{F}'_X, q}\right) \geq 4\gamma/\phi$. Otherwise, trivially, there can be at most $\tilde{O}(1)$ edges affected in $C_{\bar{F}_X', q}$. Deleting $e$ only affects the components containing its endpoints, so at most two such components can undergo such changes.
    
        Now, given the assumption, $q$ is either the centroid of the connected component, or $C_{\bar{F}_X', q, f} = T'_q$ (there is only a unique edge that can be removed such that the remaining volume is mostly in the other component, namely, the first edge on the path towards the centroid). Thus, $q \in Q$. But $Q$ is of size $\tilde{O}(1)$ and each vertex $q \in Q$ has $\deg_{\bar{G}'}(q) \leq \vol_{\bar{G}'_X}(T'_q) < 2\gamma/\phi$. Thus, there are at most $\tilde{O}(1) \cdot 2\gamma/\phi = \tilde{O}(1)$ such edges. Edges incident to the centroid only need to be considered if the centroid has degree less than $2\gamma/\phi$ in $\bar{G}'_X$, and thus this contributes at most another $\tilde{O}(1)$ edges.

        For edges $f$ that were previously not volume-covered, but are volume-covered in $\bar{F}'_X$ and $\bar{G}'_X$, an identical argument applies by reversing the updates (transitioning from $\bar{G}'_X$ back to $\bar{G}_X$), yielding at most $\tilde{O}(1)$ edges.

        \item \underline{Change in cut-covered status:} recall that edge $f$ is cut-covered in $\bar{G}_X$ if for at least one endpoint $p$ of $f$, we have $\vol_{\bar{G}_X}\left(C_{\bar{F}_X, p, f}\right) < 2\gamma/\phi$ \textbf{and} there is an edge $f' \in \partial_{\bar{G}_X} C_{F', p, f}$ and $f' \neq f$. It is thus immediate from the arguments from the volume-covered case to see that any such edge $f$ has an endpoint in $Q$ or in one of at most two centroids of degree at most $2\gamma/\phi$, and thus there are at most $\tilde{O}(1)$ such edges. The same goes for edges that are newly cut-covering.
    \end{itemize}
    Since there is at most one edge $e'$ joining $\bar{F}_X$, all other edges in $\bar{G}_X \setminus \bar{F}_X$ are also present in $\bar{G}_X' \setminus \bar{F}_X'$. The claim follows. 
\end{proof}

The following corollary is immediate from \Cref{clm:littleVolumeInSmallComponents} and the fact that a single deletion to $G$ can only affect one of the expanders $G[X]$ for $X \in \mathcal{X}$, which in turn can result in at most $\tilde{O}(1)$ deletions from $\bar{G}_X$ by \Cref{lem:expander_pruning}.

\begin{corollary}
    \label{cor:fewStatusUpdates}
    For any stage $t > 0$, at most $\tilde{O}(1)$ edges have a different covering status at time $t$ than at time $t-1$. That is, they are volume-covered or cut-covered at one time but no longer at the other, or joined $\bar{F}_X$ for some $X \in \mathcal{X}$.
\end{corollary}

Our central claim for the recourse is now implied, since only edges that change their covering status can become bridges.

\begin{corollary}
    \label{cor:NisStable}
    At any time $t$, at most $\tilde{O}(1)$ edges are removed from $N$.
\end{corollary}

We can next prove the upper bound on the number of edges in the sparsifier. 

\begin{lemma}
    $|E(H^{(0)})| \leq n - (\text{\#connected components of }G^{(0)}) + m/16$, and after each deletion to $G$, at most $\tilde{O}(1)$ new edges are added to $H$.
\end{lemma}
\begin{proof}
    Below, we write $G = G^{(0)}$, $H = H^{(0)}$, and $N_X = N_X^{(0)}$ for each cluster $X \in \mathcal{X}$. For a graph $G'$, let $c(G')$ be the number of its connected components.
    
    Clearly, $G \setminus \partial_{G} \mathcal{X} = \bigcup_{X \in \mathcal{X}} G[X]$, and by the disjointness of expanders, every bridge in some $G[X]$ is also a bridge in $G \setminus \partial_{G} \mathcal{X}$. By \Cref{inv:nonBridgeIffCovered}, for each $X \in \mathcal{X}$, the set of bridges in $G[X]$ is $G[X] \setminus N_X$. Thus, $G \setminus \partial_{G} \mathcal{X}$ contains at least $b = |\bigcup_{X \in \mathcal{X}} G[X] \setminus N_X|$ bridges. 

    Next, note that
    \[
        G \cap N = \bigcup_{X \in \mathcal{X}} G[X] \cap N_X = (G \setminus \partial_{G} \mathcal{X}) \setminus \left(\bigcup_{X \in \mathcal{X}} G[X] \setminus N_X\right).
    \]
    Since removing a bridge increases the number of components by $1$, we conclude that
    \[
        c(G \cap N) = c(G \setminus \partial_{G} \mathcal{X}) + b \geq c(G) + b.
    \]
    Thus, the number of edges in $\msf(G \cap N)$ is at most $n - (c(G) + b)$. 

    Recall that $H = (G \setminus N) \cup \msf(G \cap N)$. It is now easy to see that 
    \begin{align*}
        |E(H)| &\leq |E(G \setminus N)| + |E(\msf(G \cap N))| \\
        &\leq  |\partial_{G} \mathcal{X}| + \left|\bigcup_{X \in \mathcal{X}} G[X] \setminus N_X\right| + |E(\msf(G \cap N))| \\
        &\leq m/16 + b + n - c(G) - b = m/16 + n - c(G),
    \end{align*}
    where we use in the second inequality $G \setminus N = G \setminus ( \bigcup_{X \in \mathcal{X}} N_X) = \partial_{G} \mathcal{X} \cup (\bigcup_{X \in \mathcal{X}} G[X] \setminus N_X)$, and in the last inequality the upper bound on $|\partial_{G} \mathcal{X}|$ from our choice of $\phi$.

    Finally, by \Cref{cor:NisStable}, the set $N$ shrinks by at most $\tilde{O}(1)$ edges per deletion in $G$, and no edge removed from $N$ can be added again. Recall that $H^{(t)} = (G^{(t)} \setminus N^{(t)}) \cup \msf(G^{(t)} \cap N^{(t)})$. The above recourse bound immediately bounds the number of changes in $G^{(t)} \setminus N^{(t)}$ by $\tilde{O}(1)$, and since each deletion in a graph can cause the minimum spanning tree to change by at most one edge, $\msf(G^{(t)} \cap N^{(t)})$ also undergoes at most $\tilde{O}(1)$ changes. Thus, $H^{(t)}$ has at most $\tilde{O}(1)$ new edges present.
\end{proof}

\paragraph{Runtime Analysis.}

Finally, let us bound the runtime of the algorithm.

\begin{lemma}
The algorithm requires $\tilde{O}(m)$ time to initialize and then processes each edge update to $G$ in worst-case time $\tilde{O}(1)$. The algorithm succeeds with high probability.
\end{lemma}
\begin{proof}
By \Cref{lem:expander_decomp} and \Cref{lem:expander_pruning}, we can initialize and update all expanders in the claimed runtime. Since  \Cref{lem:expander_decomp} is randomized and only succeeds with high probability, so does our algorithm. Whenever an expander $X \in \mathcal{X}^\text{small}$ is affected, we can rerun a static algorithm to compute all bridges/non-bridges, which runs in time $\tilde{O}(1)$ since the graph is small. 

For $X \in \mathcal{X}^\text{large}$, we can update all forests $\bar{F}_X$ efficiently using the minimum spanning forest algorithm from \Cref{thm:main}. Using a link-cut tree data structure for each forest $\bar{F}_X$, we can maintain the centroid of each component and query in $O(\log n)$ time for each vertex $x$ the size of $T_x$ (if $x$ is not the centroid) and the incident volume to $T_x$. Thus, the constructive proofs from \Cref{clm:robustToMeasureChange}, \Cref{clm:robustToForestChange}, and \Cref{cor:fewStatusUpdates} allow us to find all $\tilde{O}(1)$ edges that change their covering status in time $\tilde{O}(1)$ and compute the new covering status. This allows us to maintain the set $N$ efficiently. 

Finally, to maintain the sparsifier $H$, we use the fact that by \Cref{cor:NisStable}, the set $N$ changes by at most $\tilde{O}(1)$ edges, which results in at most $\tilde{O}(1)$ changes to the graphs $G \setminus N$ and $G \cap N$. From \Cref{thm:main}, we can thus maintain the minimum spanning forest on the latter graph with $\tilde{O}(1)$ worst-case update time per time step. 
\end{proof}

%% file: coreGraphFramework.tex
\section{Fully-Dynamic Algorithms via Edge Sparsifiers}
\label{sec:mainSec}

In this section, we prove our main theorem using our novel edge sparsifiers from \Cref{sec:mst} and \Cref{sec:2edgeconn}. 

\mainThm*

We first describe in \Cref{subsec:completeMST} how to obtain a fully-dynamic minimum spanning forest algorithm by giving a framework that interleaves vertex sparsification with the edge sparsifier from \Cref{thm:edgeSparsifierMST}. We then show in \Cref{subsec:complete2EdgeConn} how to adapt this framework to get fully-dynamic $2$-edge connectivity.

\subsection{Fully-Dynamic Minimum Spanning Tree}
\label{subsec:completeMST}

To simplify the exposition, we first describe an algorithm that achieves only \emph{amortized} update time $\tilde{O}(1)$. We later discuss how to de-amortize the algorithm using standard techniques.

The algorithm maintains a hierarchy over $\Lambda + 1$ layers for $\Lambda = \ceil{\log_2 m} + 4$.
At initialization and after every update, we update the hierarchy one layer after another in increasing order of $i$, i.e., we process layer $i$ before processing layer $i+1$. We let time $t$ refer to the number of updates to $G$; that is, initialization takes place at time $0$, and the processing after the $t$-th update to $G$ takes place at time $t$. We define for $0 \leq i \leq \Lambda$, $x_i \defeq \ceil{2^{\Lambda - i - 8}/ ((\Lambda+1)\eta)}$ where $\eta = \tilde{O}(1)$ is a universal upper bound on the recourse obtained by algorithms in \Cref{thm:edgeSparsifierMST} or \Cref{thm:edgeSparsifier2conn}. Then, we \emph{(re-)initialize} layer $i$ at every time step $t$ divisible by $x_i$. That is, we remove all information stored at level $i$ and recompute it. By construction, if we (re-)initialize layer $i$, then we also (re-)initialize all layers $j > i$.

For each layer $0 \leq i \leq \Lambda$, we maintain a graph $G_i$. Intuitively, as $i$ increases, $G_i$ gets closer to a spanning forest. Formally, we maintain the following invariant.

\begin{restatable}{invariant}{keyInv}\label{inv:keyInv}
At any time, $G = G_0 \supseteq G_1 \supseteq G_2 \supseteq \cdots \supseteq G_{\Lambda}$, and for any layer $i \in [\Lambda]$, we have:
\begin{enumerate}
    \item $|E(G_i)| \leq n-\text{(\#connected components in $G$)}  + m/2^i$, and
    \item $G_i$ spans $G$, i.e., $\text{(\#connected components in $G_i$)} = \text{(\#connected components in $G$)}$.
\end{enumerate}
\end{restatable}

To this end, we use the edge sparsifiers developed in the first part of the article, and achieve vertex sparsification simply by contracting paths to off-forest edges and periodic rebuilds (it is this part that requires amortization). We have at all times $G_0 = G$ and $F_0 = (V, \emptyset)$. More concretely, at time $t$, at each layer $i = 1, \ldots, \Lambda$, 
\begin{itemize}
    \item \underline{If $t$ is divisible by $x_i$:} $F_i$ is (re-)initialized to be a  minimum spanning forest of $G_{i-1}$ (i.e., a forest spanning the connected components of $G_{i-1}$ with minimum weight). Let $\bar{E}_{i} = E(G_{i-1}) \setminus F_{i}$ be the off-forest edges. Let $V(\bar{E}_{i})$ denote the set of endpoints of $\bar{E}_{i}$. Let $B_i$ be the set of branching vertices of $V(\bar{E}_i)$ in $F_i$, i.e., the minimal set such that for every $u,v \in V(\bar{E}_i)$, their least common ancestor (we can arbitrarily assign $F_i$ a root vertex) in $F_i$ (if one exists) is in $B_i$.

    The algorithm computes the Steiner forest $\bar{F}_i$ of $F_i$ with respect to $B_i \cup V(\bar{E}_i)$, that is, the forest obtained from $F_i$ by greedily removing degree-1 vertices that are not in $B_i \cup V(\bar{E}_i)$ and then replacing maximal path segments in $F_i$ that are internally vertex-disjoint with $B_i \cup V(\bar{E}_i)$ by edges between the path's endpoints, taking the maximum weight of any edge on the path. For every edge $e \in F_i$, we denote by $\bar{e}$ the edge in $\bar{F}_i$ that was obtained from contracting the unique path containing $e$, and define $\bar{e}$ to be $\bot$ if no such path exists. For convenience, we define for $e \in \bar{E}_i$, $\bar{e} = e$.

    The algorithm then computes the reduced graph $\bar{G}_i = \bar{F}_i \cup \bar{E}_i$.  Finally, we (re-)initialize the edge sparsifier data structure $\mathcal{D}_i$ from \Cref{thm:edgeSparsifierMST} on $\bar{G}_i$ and forest $\bar{F}_i$ to maintain a sparsifier $\bar{H}_i$. Then, $G_i$ is (re-)initialized to $G_i = F_i \cup (\bar{H}_i \setminus \bar{F}_i)$. 
    
    \item \underline{Otherwise:} We distinguish by updates:
    \begin{itemize}
        \item \underline{for every edge $e$ inserted into $G_{i-1}$:} the edge $e$ is also added to $G_i$.
        \item \underline{for every edge $e$ deleted from $G_{i-1}$:} we first remove $e$ from $G_i$ if exists. Then,
        \begin{itemize}
            \item \underline{if $e \in F_i$:} then we remove it from $F_i$ and remove $\bar{e}$ from $\bar{G}_i$ (and thus $\bar{F}_i$).
            \item \underline{if $e \in G_{i-1} \setminus F_i$:} we remove $e$ from $\bar{G}_{i}$.
        \end{itemize}
        We forward the updates to $\bar{G}_i$ to $\mathcal{D}_i$ and add the additional edges added to the sparsifier $\bar{H}_i$ to $G_i$.  
    \end{itemize}
\end{itemize}

This completes the high-level overview of our algorithm. We next prove that \Cref{inv:keyInv} remains correct, and analyze the runtime of the above algorithm.

\paragraph{Correctness of the Algorithm.} We start by establishing our invariant.

\keyInv*
\begin{proof}
    We have $G_0 = G$ and $F_0 = (V, \emptyset)$ by definition. The property $G_{i-1} \supseteq G_{i}$ is straightforward from the definition of the algorithm. The fact that $G_i$ spans $G$ follows from our definition of $G_i$, $\bar{G}_i$, and the property from \Cref{thm:edgeSparsifierMST} that preserves the MSF of the input graph.

    We prove the invariant by induction on the time $t$ and the layer $i$. Let us consider any layer $i$. Let us again distinguish between (re-)initialization and normal time steps and prove slightly strengthened claims:
\begin{itemize}
    \item \underline{$t$ divisible by $x_i$:} By induction, the graph $G_{i-1}$ satisfies 
    \[|E(G_{i-1})| \leq n-\text{(\#connected components in $G_{i - 1}$)}  + m/2^{i-1}.\]
    Thus, there are at most $m/2^{i-1}$ off-forest edges $\bar{E}_i$ as any spanning forest $F_i$ of $G_{i-1}$ has $n-\text{(\#connected components in $G_{i - 1}$)}$ edges.
    
    We can therefore upper bound the number of endpoints $V(\bar{E}_i)$ trivially by $2\cdot m/2^{i-1}$, and it is well-known that the branching set of a spanning forest is at most as large as the underlying set. Thus, $|V(\bar{E}_i) \cup B_i| \leq 4 \cdot m/2^{i - 1}$. 
    
    Since $\bar{F}_i$ is a forest over vertex set $V(\bar{E}_i) \cup B_i$, it has at most $4 \cdot m/2^{i-1} - 1$ edges. Thus, $\bar{G}_i = \bar{F}_i \cup \bar{E}_i$ consists of at most $m_i = 5 \cdot m/2^{i-1}$ edges. Thus, $\bar{H}_i \setminus \bar{F}_i$ contains at most $m_i/16 < m/2^i$ edges. Since $G_i$ is set to $F_i \cup (\bar{H}_i \setminus \bar{F}_i)$, we have 
    \[
        |E(G_{i})| \leq n-\text{(\#connected components in $G_i$)} + m/2^i.
    \]
    \item \underline{Otherwise:} first, note that deleting an edge in $G_i$ increases the number of connected components by at most $1$, so this update cannot make the inequality untrue.

    If $e$ is inserted into $G$, it is simply added to every graph $G_i$ that is not (re-)initialized. Thus, there is only a single edge added to $G_i$. On the other hand, a deletion in $G_{i-1}$ results in at most one deletion in $\bar{G}_i$. By definition of $\eta$, this results in at most $\eta$ additional edges in $\bar{H}_i$, which are then added to $G_i$. But note that each deletion in $G_{i-1}$ results in only a single deletion in $G_{i}$. Thus, for $0 < j \leq i$, a deletion in $G$ causes at most a single deletion in each $G_{j-1}$, which causes $\eta$ insertions into $\bar{H}_{j}$, which are then added to $G_j$. Note that here we implicitly exploit that if $t$ is not divisible by $x_i$, then it is not divisible by $x_j$ for $j \leq i$. Thus, in either case, a single update to $G$ causes at most $(i+1)\eta \leq (\Lambda+1)\eta$ edge insertions to $G_i$.

    Since at most $t - \lfloor t / x_i\rfloor x_i \leq x_i$ time steps have passed since the last (re-)initialization of $G_i$, we conclude 
    \begin{align*}
        |E(G_i)| &\leq n-\text{(\#connected components in $G_i$)} + m_i / 16 + x_i \cdot (\Lambda+1)\eta \\&\leq n-\text{(\#connected components in $G_i$)} + m/2^{i}.
    \end{align*}
    Note that we can assume $x_i > 1$ (since otherwise, layer $i$ is (re-)initialized at every time step), which yields $x_i \cdot (\Lambda + 1)\eta \leq 2^{\Lambda - i - 7} \leq \frac{1}{4} \cdot \frac{m}{2^i}$. \qedhere 
\end{itemize}
\end{proof}

Next, we prove that each graph $G_i$ contains the minimum spanning forest of $G$.

\begin{claim}\label{clm:mstContained}
    At any stage, for any layer $i$, $\msf(G_i) = \msf(G)$.
\end{claim}

\begin{proof}
    We prove the claim by induction on the time $t$ and layer $i$. Whenever $G_i$ is (re-)initialized, we compute $F_i = \msf(G_{i-1})$ and add $F_i$ to $G_i$. Thus, clearly, at times $t$ divisible by $x_i$, we have $\msf(G_i) = F_i = \msf(G_{i - 1}) = \msf(G)$, where we use the induction hypothesis in the last equality.

    For any other time $t$, we may have two cases. If a new edge $e$ is inserted into $G^{(t-1)}$ to obtain $G^{(t)}$, we promptly add it to $G_0, G_1, \ldots, G_i$. Then
    \[
        \msf(G^{(t)}) = \msf(G^{(t-1)} \cup \{e\}) = \msf(G^{(t-1)}_i \cup \{e\}) = \msf(G^{(t)}_i),
    \]
    as desired. Otherwise, an edge $e$ is deleted from $G^{(t-1)}$ to obtain $G^{(t)}$. But note that the edges not inserted into $G_i$ are only the edges in the original set $\bar{E}_i$ that have since not been added by the data structure $\mathcal{D}_i$ to $\bar{H}_i$. From the construction of $\bar{G}_i$ and the property of \Cref{thm:edgeSparsifierMST} that we always have $\msf(\bar{G}_i) = \msf(\bar{H}_i)$, we get that, again, the minimum spanning forest of $G$ is in $G_i$.    
\end{proof}

We immediately get that the forest $F_{\Lambda}$ is a minimum spanning forest.

\begin{corollary}
At any stage, $F_{\Lambda}$ is a minimum spanning forest of $G$.
\end{corollary}
\begin{proof}
    Follows directly from the second property of \Cref{inv:keyInv} and \Cref{clm:mstContained}.
\end{proof}

\paragraph{Runtime Analysis.}

We first bound the size of instances.

\begin{claim}\label{clm:instanceSize}
    At any stage, for any layer $0 \leq i \leq \Lambda$, $G_i \setminus F_i$ contains at most $m/2^{i} + x_i$ edges.
\end{claim}

\begin{proof}
    By \Cref{inv:keyInv}, at all times we have $|E(G_i)| \leq n - c(G_i) + m/2^i$. At the most recent (re-)initialization, $F_i$ is a spanning forest of $G_{i - 1}$ (and thus spans $G_i$), so it has $n - c(G_i)$ edges where $c(G_i)$ is the number of components of $G_i$. Over the at most $x_i$ update steps since that (re-)initialization, each update to $G$ results in at most one edge deletion in $F_i$ and no insertions into $F_i$, so $|F_i| \geq n - c(G_i) - x_i$. Therefore,
    \[
        |E(G_i) \setminus F_i| = |E(G_i)| - |F_i| \leq \left(n - c(G_i) + \frac{m}{2^i}\right) - \left(n - c(G_i) - x_i\right) = \frac{m}{2^i} + x_i. \qedhere
    \]
\end{proof}

\begin{lemma}
    \label{lma:runtime}
    There is a fully-dynamic graph algorithm for the \textup{\textsc{MSF}} problem with $\tilde{O}(1)$ amortized update time. The algorithm is Las-Vegas randomized, i.e., it succeeds at each update with high probability, and reports in the case of failure. The algorithm works against an adaptive adversary.
\end{lemma}

\begin{proof}
    To implement the above framework efficiently, we maintain each forest $F_i$ in a link-cut tree data structure where each edge update requires $O(\log n)$ worst-case update time. This allows us to query the least common ancestor of any two vertices in $O(\log n)$ worst-case time.

    Let us first compute the cost of a (re-)initialization at layer $i$. We can make a copy of the link-cut data structure of $F_{i-1}$ in $\tilde{O}(1)$ worst-case time (using persistence \cite{driscoll1986making}). We can then compute the minimum spanning forest of $G_i$ as follows: we find the branching set $B$ of $V(E(G_{i-1}\setminus F_{i-1}))$, and compute the reduced forest $\bar{F}_{i}$ where degree $1$ vertices that are not in $B \cup V(E(G_{i-1}\setminus F_{i-1}))$ are greedily removed and maximal path segments that are internally vertex-disjoint to $B \cup V(E(G_{i-1}\setminus F_{i-1}))$ are contracted, taking the maximum weight of any edge on the path. We can construct this graph in $\tilde{O}(m/2^{i})$ time since $|E(G_{i-1}\setminus F_{i-1})| \leq m/2^{i}$ by \Cref{inv:keyInv} and the fact that $F_{i-1}$ is initially spanning all connected components and has undergone at most $x_{i-1} \leq m/2^{i-1}$ edge deletions (which bounds the number of connected components that it spans). The new minimum spanning forest $F_i$ can then be obtained by constructing the minimum spanning forest on the reduced graph and making appropriate changes to $F_{i-1}$. Again, this can be implemented straightforwardly in time $\tilde{O}(m/2^i)$. The same argument upper bounds the size of the reduced graph $\bar{G}_i$ that is constructed and input to the algorithm from \Cref{thm:edgeSparsifierMST}, which again only requires $\tilde{O}(m/2^i)$ initialization time. Since layer $i$ is (re-)initialized every $x_i$ time steps, the amortized update time is $\tilde{O}(m/2^i \cdot 1/x_i) = \tilde{O}(1)$, as desired.

    Further, each update at layer $i$ costs at most $\tilde{O}(1)$ time. There are only this many updates to $G_{i}$ for a single update to $G$, and each update is almost immediately forwarded to the data structure $\mathcal{D}_i$, which processes each update in time $\tilde{O}(1)$ by \Cref{thm:edgeSparsifierMST}.

    Since each edge sparsification algorithm from \Cref{thm:edgeSparsifierMST} succeeds with high probability or reports failure, so does our algorithm, as we can simply forward the reporting.
\end{proof}

\paragraph{De-amortization.} 

Lastly, we claim that \Cref{lma:runtime} can be de-amortized to achieve a worst-case $\tilde{O}(1)$ update time. Standard de-amortization techniques suggest mitigating the issue of high rebuild costs at a time step by initializing in the background, running at each stage only $\tilde{O}(1)$ operations, and only (re-)initializing the data structures once the rebuild is completed. This requires planning for delays, which causes us to use shorter rebuild periods $x_i$ to enforce the properties above. This increases the runtime only by a constant factor. We refer the reader to \cite{meierhans2026dynamic} for a more thorough discussion of such techniques and how they can be applied to this framework.

\subsection{Fully-Dynamic 2-Edge Connectivity}
\label{subsec:complete2EdgeConn}

For $2$-Edge-Connectivity, we suggest the same algorithm as presented in \Cref{subsec:completeMST} with two simple modifications:
\begin{enumerate}
    \item When (re-)initializing a layer $i$, instead of letting $F_i$ be the minimum spanning forest of $G_i$, we let $F_{i}$ be an arbitrary spanning forest of $G_{i-1}$ that extends $F_{i-1}$, i.e., $F_{i-1} \subseteq F_i$.
    \item For the edge sparsification data structure  $\mathcal{D}_i$, we take the algorithm from \Cref{thm:edgeSparsifier2conn} instead of \Cref{thm:edgeSparsifierMST}. 
\end{enumerate}
It is not hard to see that the proof of \Cref{inv:keyInv} and the runtime analysis in \Cref{lma:runtime} (and its de-amortization) are unaffected by these changes. Additionally, by (re-)initializing $F_i$ as an extension of $F_{i-1}$, we have that $F_0 = (V, \emptyset) \subseteq F_1 \subseteq F_2 \subseteq \ldots \subseteq F_{\Lambda}$ at all times.

\paragraph{Implementing $2$-Edge Connectivity Queries.}

To implement queries, we use the nestedness property of these forests: for each data structure $\mathcal{D}_i$, we obtain, along with the sparsifier $\bar{H}_i$, a set $\bar{N}_i$ of all edges that are non-bridges in $\bar{G}_i$. Recall that $\bar{G}_i = \bar{F}_i \cup \bar{E}_i$. Each edge $e =xy \in \bar{F}_i$ corresponds to a path in $F_i$, and thus in $F_{\Lambda}$ by nestedness. Recall that $\bar{E}_i \subseteq E$. Now, for each edge $\bar{e} = xy \in \bar{N}_i$:
\begin{itemize}
    \item \underline{if $\bar{e} \in \bar{F}_i$:} then we mark all edges on the $x$-to-$y$ path in $F_{\Lambda}$ as covered. 
    \item \underline{otherwise:} we have $\bar{e} \in \bar{E}_i \subseteq E$, and we simply cover $e = \bar{e}$. 
\end{itemize}

Upon query, we return that $G$ is $2$-edge connected if and only if $G$ is connected and every edge in $F_{\Lambda}$ is covered. Upon query for arbitrary vertices $u,v \in V$, we return that $u,v$ are $2$-edge connected if and only if there is a $u$-to-$v$ path in $F_{\Lambda}$ and every edge on the path is covered.

We next prove the correctness of our query implementation.

\paragraph{Correctness of the Query.} We first prove the following basic claim on the usefulness of the covering information on reduced graphs $\bar{G}_i$.

\begin{claim}\label{clm:reducedGraphCorrespondence}
For every layer $i$, and edge $e \in G_i$, we have that if $\bar{e} = \bot$, then $e$ is a bridge in $G_i$, and otherwise, $e$ is $2$-edge connected in $G_i$ if and only if $\bar{e}$ is $2$-edge connected in $\bar{G}_i$.
\end{claim}
\begin{proof}
Any edge that is incident to a degree $1$ vertex is clearly a bridge. Removing a bridge from a graph does not change the connectivity of any other edge. Thus, greedily removing degree $1$ vertices in $F_i$ that are not in $B_i \cup V(\bar{E}_i)$ (and thus also have degree $1$ in $G_{i-1}$) can only remove bridges of $G_{i-1}$.

For the remaining edges of $G_{i-1}$, some paths where all internal vertices have degree $2$ are contracted into a single edge to obtain $\bar{G}_i$. Note that the endpoints of an edge $e = xy$ are $2$-edge connected if and only if $e$ is on some cycle in $G_{i-1}$. But for a degree $2$ vertex $x$, both incident edges must be on any such cycle, and thus $x$ can be contracted into one of its neighbors. Since a one-to-one correspondence between cycles can be established, this also does not affect the $2$-edge connectivity of any other edge in the graph. The claim follows.
\end{proof}

\begin{claim}[Correctness of Covering Information]\label{clm:iCovered}
    For any layer $i$, we say that an edge $e \in E$ is $i$-covered if it is covered by any of the sets $\bar{N}_1, \bar{N}_2, \ldots, \bar{N}_i$. For every edge $e \in G_i$, if it is $i$-covered, then its endpoints are $2$-edge connected in $G$. Otherwise, it is $2$-edge connected in $G$ if and only if it is $2$-edge connected in $G_i$.
\end{claim}
\begin{proof}
We prove the claim by induction on $i$. For $i = 0$, the claim is vacuously true since $G_0 = G$. For $i > 0$, we have that $G_i$ is the pre-image of the sparsifier $\bar{H}_i$ w.r.t. $F_i$ of the reduced graph $\bar{G}_i$. By \Cref{clm:reducedGraphCorrespondence} and \Cref{thm:edgeSparsifier2conn}, the claim follows.
\end{proof}

We have, from \Cref{clm:iCovered} and $G_{\Lambda} = F_{\Lambda}$ spanning the same components as $G$ (by \Cref{inv:keyInv}), that every edge in $\bar{F}_{\Lambda}$ is covered if and only if it is $2$-edge connected in $G$. It remains to allude to the well-known fact (see, for example, \cite{frederickson1997ambivalent}) that every two vertices $x,y \in V$ are $2$-edge connected in $G$ if and only if, for a maximal spanning forest $F_{\Lambda}$ of $G$, each edge on the $x$-to-$y$ path in $F_{\Lambda}$ is $2$-edge connected. 

\paragraph{Runtime Analysis.} It remains to give a runtime analysis.

\begin{claim}[Runtime for Covering Information]
The above algorithm can be implemented to run with initialization time $\tilde{O}(m)$, worst-case update and query time $\tilde{O}(1)$. The algorithm succeeds w.h.p. and withstands an adaptive adversary.
\end{claim}
\begin{proof}
Since \Cref{lma:runtime} (and its de-amortization) extends, we only analyze the additional overhead from organizing covering information. 

To this end, we use a link-cut data structure to maintain the forest $F_{\Lambda}$ explicitly along with the covering information. That is, for every layer $i$ and edge $\bar{e} \in \bar{N}_i$, if $\bar{e} = xy \in \bar{F}_i$, we add one unit of mass to each edge on the $x$-to-$y$ path in $F_{\Lambda}$, and if it is an edge $\bar{e} = e \in F_{\Lambda}$, we simply add one unit of mass to this edge. We maintain with each edge $e \in F_{\Lambda}$ a counter $\Delta(e)$ that tracks the number of times that some unit of mass was added. Note that adding or subtracting one unit of mass to a path/edge can be done in $\tilde{O}(1)$ worst-case time via the link-cut data structure.

For queries, the link-cut data structure allows us to ask in $\tilde{O}(1)$ worst-case time for an edge $e$ in the forest with minimal mass $\Delta(e)$, or given vertices $x, y \in V$, for an edge $e$ of minimal mass $\Delta(e)$ along the $x$-to-$y$ path. In either case, if $\Delta(e) = 0$, it is not covered and we can output $e$ as a bridge; otherwise, we confirm $2$-edge connectivity. 

Finally, to update the mass $\Delta(e)$ in $F_{\Lambda}$ efficiently, we use the following strategy: we let all other data structures process the current update to $G$; then we learn about all edges $\bar{e} = xy$ that are removed from $\bar{N}_i$ on every layer $i$. There are only $\tilde{O}(1)$ such edges by \Cref{thm:edgeSparsifier2conn}, and we can remove the mass from all the corresponding paths/edges in $F_{\Lambda}$. We then process the updates to $F_{\Lambda}$; there should be only one such edge change. Finally, we add mass along all $x$-to-$y$ paths for edges $\bar{e} = xy$ that joined $\bar{N}_i$ at the current stage.

The above update strategy enforces that every update is along a tree path in $F_{\Lambda}$ and can thus be processed in $\tilde{O}(1)$ worst-case update time, as desired.
\end{proof}

%% file: appendix.tex
\section{Proofs of the centroid properties}
\label{app:centroid}

In this section, we prove two claims about centroids, which are used in \Cref{subsubsec:recourse} for the recourse analysis of our $2$-edge connectivity sparsifier.

\robustToMeasureChange*

\begin{proof}
    Consider first the case where the measure is decreased in a coordinate $u$, i.e., $\mu'(u) = \mu(u) - 1$. Clearly, only vertices from $C_{F', u}$ can be in $Q$ since all other components have no change to the incident measure. Further, if $\mu'(C_{F', u}) \leq 4\gamma/\phi + 3 = \tilde{O}(1)$, then the claim is vacuously true since every vertex in $C_{F', u}$ is incident to at least one edge in $F'$ and $\mu'$ dominates $\deg_{F'}$. Let us therefore assume that $\mu'(C_{F', u}) > 4\gamma/\phi + 3$. 

    Next, we argue that for any vertex $q$ with $\mu'(T_{F', \mu', q}) \leq 2\gamma/\phi$, we have $T_{F', \mu', q} = T_{F', \mu, q}$. This follows because by assumption $\mu(T_{F', \mu', q}) \leq \mu'(T_{F', \mu', q}) + 1 \leq 2\gamma/\phi + 1$ and since $\mu(C_{F', u}) \geq \mu'(C_{F', u}) - 1 > 4\gamma/\phi + 2$. Thus, $\mu(T_{F', \mu', q}) \leq \mu(C_{F', q})/2$. But this implies that the first edge on the $x$-to-centroid path in $F'$ coincides, and thus $T_{F', \mu', q} = T_{F', \mu, q}$. 

    It follows that for all $q \in Q$, $T_{F', \mu', q} = T_{F', \mu, q}$, and so we must have $\mu'(T_{F', \mu', q}) \neq \mu(T_{F', \mu', q})$. We subsequently exploit that the trees $T_{F', \mu', q}$ for $q \neq c_{F', \mu'}(q)$ form a laminar family. This yields that all vertices $q \in Q$ are on the $u$-to-$c_{F', \mu'}(u)$ path $u = q_1, q_2, \ldots, q_k = c_{F', \mu'}(u)$ since otherwise $u \not\in T_{F', \mu', q}$ and thus the incident measure is not changed. But note that we have $q_1, \ldots, q_i$ in $T_{F', \mu', q_i}$ and since each vertex $q_i$ for $i > 0$ has degree at least $1$ in $F'$, and $\mu'(q_i) = \mu(q_i) \geq \deg_{F'}(q_i) \geq 1$, we have that for $i > 2\gamma/\phi + 1$, $\mu'(T_{F', \mu', q_i}) > 2\gamma/\phi$ and thus $q_i \not\in Q$. It follows that at most $2\gamma/\phi+1$ vertices can be in $Q$ and since $2\gamma/\phi = \tilde{O}(1)$. See \Cref{fig:volum_dec} for an illustration of this concept. 

    The reasoning for the case where the measure increases in a single vertex is almost identical.
\end{proof}

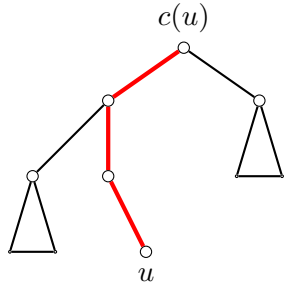
\begin{figure}[ht]
\centering
\begin{tikzpicture}[
    scale=1,
    every node/.style={circle, draw, inner sep=1.5pt},
    tree edge/.style={thick},
    subtree node/.style={circle, draw, inner sep=0pt},
]

\node (x) at (0,-0.3) {};
\node (x1) at (-1,-1) {};
\node (x2) at (1,-1) {};
\node (y1) at (-2,-2) {};
\node (z5) at (-.5,-3) {};
\node (y2) at (-1,-2) {};
\node[subtree node] (z1) at (-2.3,-3) {};
\node[subtree node] (z2) at (-1.7,-3) {};

\node[subtree node] (z3) at (-2.3+3,-2) {};
\node[subtree node] (z4) at (-1.7+3,-2) {};

\draw[ultra thick, red] (x)--(x1);
\draw[tree edge] (x)--(x2);
\draw[tree edge] (x1)--(y1);
\draw[ultra thick, red] (x1)--(y2);
\draw[ultra thick, red] (z5)--(y2);

\draw[tree edge] (z1)--(z2);
\draw[tree edge] (y1)--(z2);
\draw[tree edge] (y1)--(z1);

\draw[tree edge] (z3)--(z4);
\draw[tree edge] (x2)--(z4);
\draw[tree edge] (x2)--(z3);

\node[draw=none] at (0, 0.1) {$c(u)$};
\node[draw=none] at (-.5,-3.3) {$u$};

\end{tikzpicture}
\caption{In the proof of \Cref{clm:robustToMeasureChange}, we argue that only trees $T_q$ are affected for $q$ on the $u$-to-$c(u)$ path. But since these trees are contained in each other, and since they contain all vertices (and their subtrees) that are closer to $u$ by construction, after distance $2\gamma/\phi$ from $u$, the trees become too heavy for their roots to belong to $Q$.}
\label{fig:volum_dec}
\end{figure}

\robustToForestChange*

\begin{proof}
    Let us first consider a deletion of an edge $f = uv$ from $F'$. Clearly, $u$ and $v$ are in the same connected component of $F'$, which is then partitioned into two components in $F''$, one containing $u$ and one containing $v$. Any vertex in a connected component $C \neq C_{F', u}$ is not incident to any change and cannot be in $Q$. For the rest of the proof, let us bound the number of vertices from $C_{F'', u}$ that are in $Q$ by $\tilde{O}(1)$. By symmetry, the same bound then extends to $C_{F'', v}$ and establishes the overall claim.

    Next, consider any vertex $q \in C_{F'', u}$ that is not on the $u$-to-$c_{F'', \mu}(u)$ path. It is not hard to see that then $T_{F'', \mu, q} = T_{F', \mu, q}$ and $u \not\in T_{F'', \mu, q}$, so the vertex $q$ cannot be in $Q$. Let us next consider the vertices $q_i$ on the $u$-to-$c_{F'', \mu}(u)$ path $u = q_1, q_2, \ldots, q_k = c_{F'', \mu}(u)$. Since the trees $T_{F'', \mu, q}$ for $q \neq c_{F'', \mu}(q)$ form a laminar family, we again have that for any $i$, $q_1, q_2, \ldots, q_i$ are contained in $T_{F'', \mu, q_i}$. But for $i \geq 2\gamma/\phi$,
    \[
        \mu(T_{F'', \mu, q_i}) \geq \mu(\{q_1, q_2, \ldots, q_i\}) \geq \deg_{F'}(\{q_1, q_2, \ldots, q_i\}) \geq i \geq 2\gamma/\phi
    \]
    and thus $q_i \not\in Q$. Thus, only the first $2\gamma/\phi$ vertices on the path can be contained in $Q$. Thus, at most $2\gamma/\phi = \tilde{O}(1)$ vertices from $C_{F'', u}$ can be in $Q$, as desired.

    The same reasoning almost seamlessly extends to arguing about inserting an edge into $F'$.
\end{proof}

%% file: refs.bib
@article{frederickson1985data,
  title={Data structures for on-line updating of minimum spanning trees, with applications},
  author={Frederickson, Greg N},
  journal={SIAM Journal on Computing},
  volume={14},
  number={4},
  pages={781--798},
  year={1985},
  publisher={SIAM}
}

@inproceedings{driscoll1986making,
  title={Making data structures persistent},
  author={Driscoll, James R and Sarnak, Neil and Sleator, Daniel Dominic and Tarjan, Robert Endre},
  booktitle={Proceedings of the eighteenth annual ACM symposium on Theory of computing},
  pages={109--121},
  year={1986}
}

@inproceedings{SW19,
author = {Thatchaphol Saranurak and Di Wang},
title = {Expander Decomposition and Pruning: Faster, Stronger, and Simpler},
booktitle = {Proceedings of the 2019 Annual ACM-SIAM Symposium on Discrete Algorithms (SODA)},
chapter = {},
year = {2019},
pages = {2616-2635},
doi = {10.1137/1.9781611975482.162},
URL = {https://epubs.siam.org/doi/abs/10.1137/1.9781611975482.162},
eprint = {https://epubs.siam.org/doi/pdf/10.1137/1.9781611975482.162}
}

@inproceedings{patrascu2007planning,
  title={Planning for fast connectivity updates},
  author={Patrascu, Mihai and Thorup, Mikkel},
  booktitle={48th Annual IEEE Symposium on Foundations of Computer Science (FOCS'07)},
  pages={263--271},
  year={2007},
  organization={IEEE}
}

@article{spira1975finding,
  title={On finding and updating spanning trees and shortest paths},
  author={Spira, Philip M. and Pan, A},
  journal={SIAM Journal on Computing},
  volume={4},
  number={3},
  pages={375--380},
  year={1975},
  publisher={SIAM}
}

@inproceedings{meierhans2026dynamic,
  title={Dynamic Connectivity with Expected Polylogarithmic Worst-Case Update Time},
  author={Meierhans, Simon and Gutenberg, Maximilian Probst},
  booktitle={Proceedings of the 2026 Annual ACM-SIAM Symposium on Discrete Algorithms (SODA)},
  pages={6184--6198},
  year={2026},
  organization={SIAM}
}

@article{chin1978algorithms,
  title={Algorithms for updating minimal spanning trees},
  author={Chin, Francis and Houck, David},
  journal={Journal of Computer and System Sciences},
  volume={16},
  number={3},
  pages={333--344},
  year={1978},
  publisher={Elsevier}
}

@inproceedings{chuzhoy2020deterministic,
  title={A deterministic algorithm for balanced cut with applications to dynamic connectivity, flows, and beyond},
  author={Chuzhoy, Julia and Gao, Yu and Li, Jason and Nanongkai, Danupon and Peng, Richard and Saranurak, Thatchaphol},
  booktitle={2020 IEEE 61st Annual Symposium on Foundations of Computer Science (FOCS)},
  pages={1158--1167},
  year={2020},
  organization={IEEE}
}

@inproceedings{kapron,
  title={Dynamic graph connectivity in polylogarithmic worst case time},
  author={Kapron, Bruce M and King, Valerie and Mountjoy, Ben},
  booktitle={Proceedings of the twenty-fourth annual ACM-SIAM symposium on Discrete algorithms},
  pages={1131--1142},
  year={2013},
  organization={SIAM}
}

@INPROCEEDINGS{NSW17,
  author={Nanongkai, Danupon and Saranurak, Thatchaphol and Wulff-Nilsen, Christian},
  booktitle={2017 IEEE 58th Annual Symposium on Foundations of Computer Science (FOCS)}, 
  title={Dynamic Minimum Spanning Forest with Subpolynomial Worst-Case Update Time}, 
  year={2017},
  volume={},
  number={},
  pages={950-961},
  doi={10.1109/FOCS.2017.92}}

@article{holm01, author = {Holm, Jacob and de Lichtenberg, Kristian and Thorup, Mikkel}, title = {Poly-logarithmic deterministic fully-dynamic algorithms for connectivity, minimum spanning tree, 2-edge, and biconnectivity}, year = {2001}, issue_date = {July 2001}, publisher = {Association for Computing Machinery}, address = {New York, NY, USA}, volume = {48}, number = {4}, issn = {0004-5411}, url = {https://doi.org/10.1145/502090.502095}, doi = {10.1145/502090.502095}, journal = {J. ACM}, month = jul, pages = {723–760}, numpages = {38} }

@inproceedings{HenzingerKing95, author = {Henzinger, Monika Rauch and King, Valerie}, title = {Randomized dynamic graph algorithms with polylogarithmic time per operation}, year = {1995}, isbn = {0897917189}, publisher = {Association for Computing Machinery}, address = {New York, NY, USA}, url = {https://doi.org/10.1145/225058.225269}, doi = {10.1145/225058.225269}, booktitle = {Proceedings of the Twenty-Seventh Annual ACM Symposium on Theory of Computing}, pages = {519–527}, numpages = {9}, location = {Las Vegas, Nevada, USA}, series = {STOC '95} }

@article{it97, author = {Eppstein, David and Galil, Zvi and Italiano, Giuseppe F. and Nissenzweig, Amnon}, title = {Sparsification—a technique for speeding up dynamic graph algorithms}, year = {1997}, issue_date = {Sept. 1997}, publisher = {Association for Computing Machinery}, address = {New York, NY, USA}, volume = {44}, number = {5}, issn = {0004-5411}, url = {https://doi.org/10.1145/265910.265914}, doi = {10.1145/265910.265914}, journal = {J. ACM}, month = sep, pages = {669–696}, numpages = {28} }

@inproceedings{NS17, author = {Nanongkai, Danupon and Saranurak, Thatchaphol}, title = {Dynamic spanning forest with worst-case update time: adaptive, Las Vegas, and $O(n^{1/2} - \epsilon)$-time}, year = {2017}, isbn = {9781450345286}, publisher = {Association for Computing Machinery}, address = {New York, NY, USA}, url = {https://doi.org/10.1145/3055399.3055447}, doi = {10.1145/3055399.3055447}, booktitle = {Proceedings of the 49th Annual ACM SIGACT Symposium on Theory of Computing}, pages = {1122–1129}, numpages = {8}, location = {Montreal, Canada}, series = {STOC 2017} }

@inproceedings{W17, author = {Wulff-Nilsen, Christian}, title = {Fully-dynamic minimum spanning forest with improved worst-case update time}, year = {2017}, isbn = {9781450345286}, publisher = {Association for Computing Machinery}, address = {New York, NY, USA}, url = {https://doi.org/10.1145/3055399.3055415}, doi = {10.1145/3055399.3055415}, booktitle = {Proceedings of the 49th Annual ACM SIGACT Symposium on Theory of Computing}, pages = {1130–1143}, numpages = {14}, location = {Montreal, Canada}, series = {STOC 2017} }

@article{frederickson1997ambivalent,
author = {Frederickson, Greg N.},
title = {Ambivalent Data Structures for Dynamic 2-Edge-Connectivity and $k$ Smallest Spanning Trees},
journal = {SIAM Journal on Computing},
volume = {26},
number = {2},
pages = {484-538},
year = {1997},
doi = {10.1137/S0097539792226825},

URL = { 
    
        https://doi.org/10.1137/S0097539792226825
    
    

},
eprint = { 
    
        https://doi.org/10.1137/S0097539792226825
    
    

}
}

@inproceedings{expander_hierarchy,
  title={The expander hierarchy and its applications to dynamic graph algorithms},
  author={Goranci, Gramoz and R{\"a}cke, Harald and Saranurak, Thatchaphol and Tan, Zihan},
  booktitle={Proceedings of the 2021 ACM-SIAM Symposium on Discrete Algorithms (SODA)},
  pages={2212--2228},
  year={2021},
  organization={SIAM}
}

@inproceedings{kejlberg2016faster,
  title={Faster Worst Case Deterministic Dynamic Connectivity},
  author={Kejlberg-Rasmussen, Casper and Kopelowitz, Tsvi and Pettie, Seth and Thorup, Mikkel},
  booktitle={24th Annual European Symposium on Algorithms (ESA 2016)},
  pages={53--1},
  year={2016},
  organization={Schloss Dagstuhl--Leibniz-Zentrum f{\"u}r Informatik}
}

@inproceedings{holm2015faster,
  title={Faster fully-dynamic minimum spanning forest},
  author={Holm, Jacob and Rotenberg, Eva and Wulff-Nilsen, Christian},
  booktitle={Algorithms-ESA 2015: 23rd Annual European Symposium, Patras, Greece, September 14-16, 2015, Proceedings},
  pages={742--753},
  year={2015},
  organization={Springer}
}

@article{nagamochi1992linear,
  title={A linear-time algorithm for finding a sparse k-connected spanning subgraph of ak-connected graph},
  author={Nagamochi, Hiroshi and Ibaraki, Toshihide},
  journal={Algorithmica},
  volume={7},
  number={1},
  pages={583--596},
  year={1992},
  publisher={Springer}
}

@inproceedings{galil1991fully,
  title={Fully dynamic algorithms for edge connectivity problems},
  author={Galil, Zvi and Italiano, Giuseppe F},
  booktitle={Proceedings of the twenty-third Annual ACM Symposium on Theory of Computing},
  pages={317--327},
  year={1991}
}

@inproceedings{thorup2000near,
  title={Near-optimal fully-dynamic graph connectivity},
  author={Thorup, Mikkel},
  booktitle={Proceedings of the thirty-second annual ACM symposium on Theory of computing},
  pages={343--350},
  year={2000}
}

@inproceedings{holm2025fully,
  title={Fully dynamic biconnectivity in  $\tilde{O}(\log^2 n)$ time},
  author={Holm, Jacob and Nadara, Wojciech and Rotenberg, Eva and Soko{\l{}}owski, Marek},
  booktitle={Proceedings of the 57th Annual ACM Symposium on Theory of Computing},
  pages={156--165},
  year={2025}
}

@inproceedings{henzinger1995certificates,
  title={Certificates and fast algorithms for biconnectivity in fully-dynamic graphs},
  author={Henzinger, Monika R and La Poutr{\'e}, Han},
  booktitle={European Symposium on Algorithms},
  pages={171--184},
  year={1995},
  organization={Springer}
}

@article{henzinger2000improved,
  title={Improved data structures for fully dynamic biconnectivity},
  author={Henzinger, Monika R},
  journal={SIAM Journal on Computing},
  volume={29},
  number={6},
  pages={1761--1815},
  year={2000},
  publisher={SIAM}
}

@article{henzinger1995fully,
  title={Fully dynamic biconnectivity in graphs},
  author={Henzinger, Monika Rauch},
  journal={Algorithmica},
  volume={13},
  number={6},
  pages={503--538},
  year={1995},
  publisher={Springer}
}

@inproceedings{el2025fully,
  title={Fully dynamic approximate minimum cut in subpolynomial time per operation},
  author={El-Hayek, Antoine and Henzinger, Monika and Li, Jason},
  booktitle={Proceedings of the 2025 Annual ACM-SIAM Symposium on Discrete Algorithms (SODA)},
  pages={750--784},
  year={2025},
  organization={SIAM}
}

@inproceedings{de2025tree,
  title={Tree-packing revisited: Faster fully dynamic min-cut and arboricity},
  author={De Vos, Tijn and Christiansen, Aleksander BG},
  booktitle={Proceedings of the 2025 Annual ACM-SIAM Symposium on Discrete Algorithms (SODA)},
  pages={700--749},
  year={2025},
  organization={SIAM}
}

@inproceedings{jin2024fully,
  title={Fully dynamic min-cut of superconstant size in subpolynomial time},
  author={Jin, Wenyu and Sun, Xiaorui and Thorup, Mikkel},
  booktitle={Proceedings of the 2024 Annual ACM-SIAM Symposium on Discrete Algorithms (SODA)},
  pages={2999--3026},
  year={2024},
  organization={SIAM}
}

@article{thorup2007fully,
  title={Fully-dynamic min-cut},
  author={Thorup, Mikkel},
  journal={Combinatorica},
  volume={27},
  number={1},
  pages={91--127},
  year={2007},
  publisher={Springer}
}

@article{galil1992fully,
  title={Fully dynamic algorithms for 2-edge connectivity},
  author={Galil, Zvi and Italiano, Giuseppe F},
  journal={SIAM Journal on Computing},
  volume={21},
  number={6},
  pages={1047--1069},
  year={1992},
  publisher={SIAM}
}

@inproceedings{el2026deterministic,
  title={Deterministic and Exact Fully-dynamic Minimum Cut of Superpolylogarithmic Size in Subpolynomial Time},
  author={El-Hayek, Antoine and Henzinger, Monika and Li, Jason},
  booktitle={Proceedings of the 2026 Annual ACM-SIAM Symposium on Discrete Algorithms (SODA)},
  pages={613--663},
  year={2026},
  organization={SIAM}
}

@inproceedings{chalermsook2021vertex,
  title={Vertex sparsification for edge connectivity},
  author={Chalermsook, Parinya and Das, Syamantak and Kook, Yunbum and Laekhanukit, Bundit and Liu, Yang P and Peng, Richard and Sellke, Mark and Vaz, Daniel},
  booktitle={Proceedings of the 2021 ACM-SIAM Symposium on Discrete Algorithms (SODA)},
  pages={1206--1225},
  year={2021},
  organization={SIAM}
}

@inproceedings{chen23simple,
author = {Li Chen and Rasmus Kyng and Maximilian Probst Gutenberg and Sushant Sachdeva},
title = {A Simple Framework for Finding Balanced Sparse Cuts via {APSP}},
booktitle = {2023 Symposium on Simplicity in Algorithms (SOSA)},
chapter = {},
year = {2023},
pages = {42-55},
doi = {10.1137/1.9781611977585.ch5},
URL = {https://epubs.siam.org/doi/abs/10.1137/1.9781611977585.ch5},
eprint = {https://epubs.siam.org/doi/pdf/10.1137/1.9781611977585.ch5}
}

@inproceedings{meierhans2025expanderpruningpolylogarithmicworstcase,
      title={Expander Pruning with Polylogarithmic Worst-Case Recourse and Update Time}, 
      author={Simon Meierhans and Maximilian Probst Gutenberg and Thatchaphol Saranurak},
      year={2025},
    booktitle = {Proceedings of the 2026 Annual ACM-SIAM Symposium on Discrete Algorithms (SODA)}
}

@article{harel1982line,
  title={On line maintenance of the connected components of dynamic graphs},
  author={Harel, Dov},
  journal={Unpublished manuscript},
  volume={5},
  pages={95--97},
  year={1982}
}

@misc{gibb2015dynamicgraphconnectivityimproved,
      title={Dynamic graph connectivity with improved worst case update time and sublinear space}, 
      author={David Gibb and Bruce Kapron and Valerie King and Nolan Thorn},
      year={2015},
      eprint={1509.06464},
      archivePrefix={arXiv},
      primaryClass={cs.DS},
      url={https://arxiv.org/abs/1509.06464}, 
}

@misc{wang2015improvedrandomizeddatastructure,
      title={An Improved Randomized Data Structure for Dynamic Graph Connectivity}, 
      author={Zhengyu Wang},
      year={2015},
      eprint={1510.04590},
      archivePrefix={arXiv},
      primaryClass={cs.DS},
      url={https://arxiv.org/abs/1510.04590}, 
}

@article{huang2023fully,
    title      = {Fully Dynamic Connectivity in {$O(\log n(\log\log n)^2)$} Amortized Expected Time},
    author     = {Shang-En Huang and Dawei Huang and Tsvi Kopelowitz and Seth Pettie and Mikkel Thorup},
    url        = {https://theoretics.episciences.org/9645},
    doi        = {10.46298/theoretics.23.6},
    journal    = {TheoretiCS},
    issn       = {2751-4838},
    volume     = {Volume 2},
    eid        = 6,
    year       = {2023},
    month      = {May},
}

@inproceedings{WulffNilsen13, author = {Wulff-Nilsen, Christian}, title = {Faster deterministic fully-dynamic graph connectivity}, year = {2013}, isbn = {9781611972511}, publisher = {Society for Industrial and Applied Mathematics}, address = {USA}, booktitle = {Proceedings of the Twenty-Fourth Annual ACM-SIAM Symposium on Discrete Algorithms}, pages = {1757–1769}, numpages = {13}, location = {New Orleans, Louisiana}, series = {SODA '13} }

@inproceedings{eppstein90planarmst, author = {Eppstein, David and Italiano, Giuseppe F. and Tamassia, Roberto and Tarjan, Robert E. and Westbrook, Jeffery and Yung, Moti}, title = {Maintenance of a minimum spanning forest in a dynamic planar graph}, year = {1990}, isbn = {0898712513}, publisher = {Society for Industrial and Applied Mathematics}, address = {USA}, booktitle = {Proceedings of the First Annual ACM-SIAM Symposium on Discrete Algorithms}, pages = {1–11}, numpages = {11}, location = {San Francisco, California, USA}, series = {SODA '90} }

@article{EPPSTEIN1994237,
title = {Offline Algorithms for Dynamic Minimum Spanning Tree Problems},
journal = {Journal of Algorithms},
volume = {17},
number = {2},
pages = {237-250},
year = {1994},
issn = {0196-6774},
doi = {https://doi.org/10.1006/jagm.1994.1033},
url = {https://www.sciencedirect.com/science/article/pii/S0196677484710339},
author = {D. Eppstein}
}

@article{westbrook1992bridge,
	author = {Westbrook, Jeffery and Tarjan, Robert E.},
	date = {1992/06/01},
	doi = {10.1007/BF01758773},
	id = {Westbrook1992},
	isbn = {1432-0541},
	journal = {Algorithmica},
	number = {1},
	pages = {433--464},
	title = {Maintaining bridge-connected and biconnected components on-line},
	url = {https://doi.org/10.1007/BF01758773},
	volume = {7},
	year = {1992}}

@inproceedings{holm2017bridge,
author = {Jacob Holm and Eva Rotenberg and Mikkel Thorup},
title = {Dynamic Bridge-Finding in $\tilde{O}(\log^2 n)$ Amortized Time},
booktitle = {Proceedings of the 2018 Annual ACM-SIAM Symposium on Discrete Algorithms (SODA)},
year={2018},
chapter = {},
pages = {35-52},
doi = {10.1137/1.9781611975031.3},
URL = {https://epubs.siam.org/doi/abs/10.1137/1.9781611975031.3},
eprint = {https://epubs.siam.org/doi/pdf/10.1137/1.9781611975031.3}
}

@INPROCEEDINGS{jin2021fully,
  author={Jin, Wenyu and Sun, Xiaorui},
  booktitle={2021 IEEE 62nd Annual Symposium on Foundations of Computer Science (FOCS)}, 
  title={Fully Dynamic $s$-$t$ Edge Connectivity in Subpolynomial Time (Extended Abstract)}, 
  year={2022},
  volume={},
  number={},
  pages={861-872},
  doi={10.1109/FOCS52979.2021.00088}}
